\documentclass[journal, 10pt, twocolumn]{IEEEtran}

\usepackage[usenames, dvipsnames]{color}
\usepackage{citesort}
\usepackage{pifont}
\newcommand{\cmark}{\ding{51}}%
\ifCLASSINFOpdf
  \usepackage[pdftex]{graphicx}
	\usepackage{epstopdf}
\else
  \usepackage[dvips]{graphicx}
\fi
\usepackage[cmex10]{amsmath}
\usepackage{algorithmic}
\usepackage{amssymb}
\usepackage{array}
\usepackage{mdwmath}
\usepackage{mdwtab}
\usepackage{eqparbox}
\usepackage{multirow}
\usepackage{fixltx2e}
\usepackage{enumerate}
\usepackage{url}
\usepackage{tabularx}
\usepackage{verbatim}
\usepackage{color}

\newcommand{\ie}{\textit{i.e., }}
\newcommand{\eg}{\textit{e.g., }}

\newcommand{\Rmb}{\mathbb{R}}

\newcommand{\mbf}[1]{\boldsymbol{#1}}
\newcommand{\mbb}[1]{\mathbb{#1}}
\newcommand{\hl}[1]{\textcolor{black}{#1}}

\newcommand{\xb}{\mbf{x}}

\newcommand{\Xb}{\mbf{X}}
\newcommand{\Zb}{\mbf{Z}}
\newcommand{\Xbh}{\hat{\mbf{X}}}
\newcommand{\Ab}{\mbf{A}}

\newcommand{\yb}{\mbf{y}}
\newcommand{\Yb}{\mbf{Y}}

\newcommand{\eb}{\mbf{e}}
\newcommand{\Db}{\mbf{D}}
\newcommand{\Dbh}{\hat{\mbf{D}}}
\newcommand{\Ub}{\mbf{U}}

\newcommand{\Sbb}{\mathbb{S}}
\newcommand{\Ubb}{\mathbb{U}}

\newcommand{\Ib}{\mbf{I}}

\newcommand{\Wb}{\mbf{W}}

\newcommand{\rank}{\operatorname{rank}}

\newcommand{\diag}{\operatorname{diag}}
\newcommand{\argm}{\operatorname{argmin}}

\newcommand{\norm}[1]{\|{#1}\|}

\newcommand{\Rf}{\mathfrak{R}}

\newtheorem{theorem}{Theorem}
\newtheorem{lemma}{Lemma}

\newtheorem{conj}{Conjecture}
\newtheorem{proposition}{Proposition}
\newtheorem{definition}{Definition}

\begin{document}
%
\title{A Set-Theoretic Study of the Relationships of Image Models and Priors for Restoration Problems}
\author{Bihan~Wen, ~\IEEEmembership{Member,~IEEE,}~Yanjun~Li, ~\IEEEmembership{Member,~IEEE,}~Yuqi~Li, ~\IEEEmembership{Student Member,~IEEE,} and~Yoram~Bresler,~\IEEEmembership{Life Fellow,~IEEE}
\thanks{This work was supported in part by the National Science Foundation (NSF) under grants CCF-1320953 and IIS 14-47879. Bihan Wen was supported in part by Ministry of Education, Republic of Singapore, under the start-up grant.}
\thanks{B. Wen is with the School of Electrical and Electronic Engineering, Nanyang Technological University, 639798 Singapore e-mail: bihan.wen@ntu.edu.sg}
\thanks{Y. Li, Y. Li, and Y. Bresler are with the Department of Electrical and Computer Engineering and the Coordinated Science Laboratory, University of Illinois, Urbana-Champaign, IL, 61801 USA e-mail: (yli145, yuqil3, ybresler)@illinois.edu.}
}

\maketitle

\begin{abstract}
Image prior modeling is the key issue in image recovery, computational imaging, compresses sensing, and other inverse problems.
Recent algorithms combining multiple effective priors such as the sparse or low-rank models, have demonstrated superior performance in various applications.
However, the relationships among the popular image models are unclear, and no theory in general is available to demonstrate their connections.
In this paper, we present a theoretical analysis on the image models, to bridge the gap between applications and image prior understanding, including sparsity, group-wise sparsity, joint sparsity, and low-rankness, etc.
We systematically study how effective each image model is for image restoration.
Furthermore, we relate the denoising performance improvement by combining multiple models, to the image model relationships.
Extensive experiments are conducted to compare the denoising results which are consistent with our analysis.
On top of the model-based methods, we quantitatively demonstrate the image properties that are inexplicitly exploited by deep learning method, of which can further boost the denoising performance by combining with its complementary image models.
\end{abstract}


\IEEEpeerreviewmaketitle
\section{Introduction} \label{sec1}
																																												
\hl{Image restoration (IR)} aims to recover an image $\mbf{x}$ from its degraded measurements $\mbf{y}$, which can be represented as 
\begin{equation} \label{eq:ip}
\mbf{y} = \Ab\, \mbf{x} + \eb.
\end{equation}
Here $\Ab$ and $\eb$ denote the sensing operator and additive noise, respectively.
Different forms of $\Ab$ in (\ref{eq:ip}) defines a wide range of IR problems, e.g., in image denoising, $\Ab = \Ib$ with $\Ib$ denotes an identity matrix.
Furthermore, modern imaging applications usually recover high-quality $\mbf{x}$ from imcomplete or corrupted measurements $\mbf{y}$, in order to reduce the data-acquisition time (e.g., magnetic resonance imaging~\cite{lustig2008compressed}) or radiation dose (e.g., computed tomography~\cite{sidky2008image}).
Under such settings, IR becomes an ill-posed inverse problem, i.e., the unique solution $\xb$ cannot be obtained by directly inverting the linear system.
Thus, having an effective regularizer is key to a successful IR algorithm.

\hl{Popular IR methods apply regularizers by exploiting image priors, e.g., sparsity, low-rankness, etc.}
Natural images are known to be \textit{sparse}, i.e., image patches are typically sparsifiable or compressible under certain transforms, or over certain dictionaries.
Early works exploit image sparsity in fixed transform domains \cite{Dabov2007,yu2011dct,chang2000adaptive}. 
More recent IR works proposed to adapt the sparse models to image patches via data-driven approaches, such as dictionary learning \cite{elad,elad2,Mairal2009} or transform learning \cite{sabres3,octobos,wen2017sparsity}. 
They demonstrated promising performance in various inverse problems \cite{elad,elad2,Mairal2009,Dong2011,sabres3,octobos,wen2017sparsity}.
Besides local sparsity, when modeling images with diverse textures, some IR methods proposed to first group, or partition the image patches into groups of similar ones using block matching, or clustering techniques, respectively \cite{Zhang2014,octobos,Mairal2009,zha2018group}. 
They are approximately sparse \cite{Zhang2014,octobos,zha2018group}, or jointly sparse \cite{Mairal2009} under a group-based sparse model.
Apart from sparsity priors, many popular algorithms also apply low-rank modeling for each group of patches, to exploit image self-similarity \cite{Dong2013,Zhang2014,gu2017weighted,Yoon2014,zha2017image}.

\begin{figure}[!t]
\begin{center}
\begin{tabular}{c}
\hspace{-0.1in}
\vspace{-0.1in}
\includegraphics[height=2.0in]{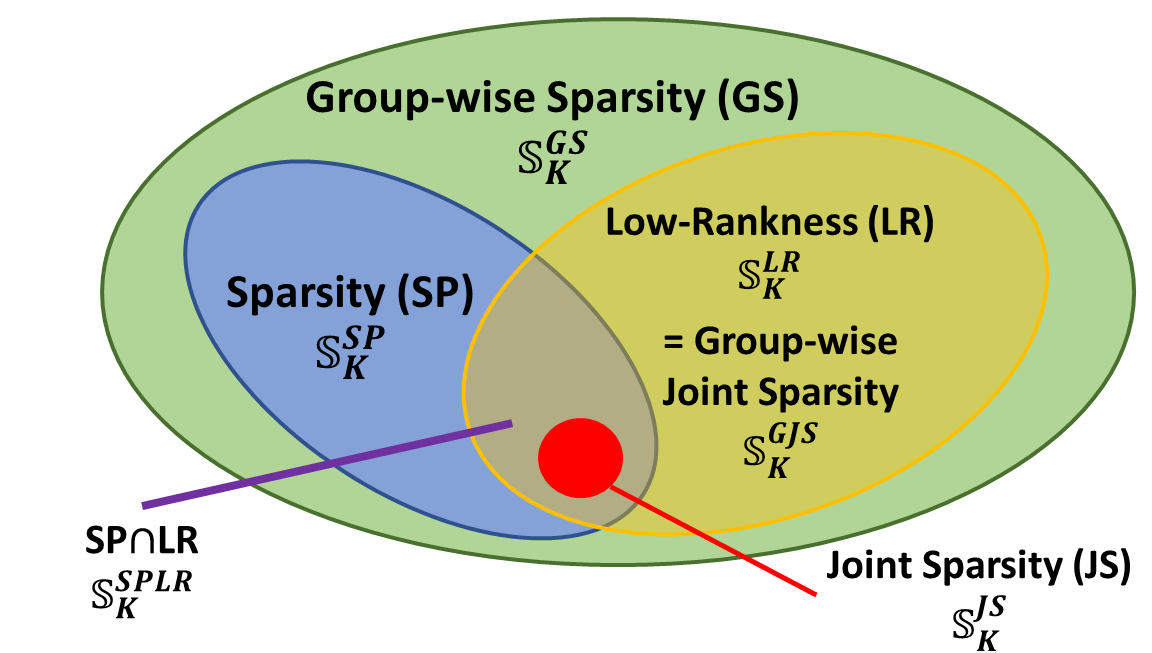}\\
\end{tabular}
\vspace{0.05in}
\caption{A Venn diagram illustrating the relationship among the sets of GS, SP, LR, GJS, JS, and SPLR image models, when the number of groups $N > 1$.}
\label{fig:Relationship}
\end{center}
\vspace{-0.2in}
\end{figure}

\hl{While numerous IR algorithms focused on exploiting single image models, very recent works \cite{zha2017analyzing,zha2018acml,wen2017joint,wen2017sparsity} proposed to jointly utilize multiple complementary models (e.g., sparsity and low-rankness), and demonstrated superior results in IR. 
Besides methods based on parsimonious models, recent deep learning approaches \cite{wang2018non,liu2018non} also combine local operations with non-local structures (that are conjectured to exploit the complementary image properties), which leads to state-of-the-art performances in a wide range of IR and computer vision applications.
Such performance improvements in practice have raised the following questions that need to be answered:
\begin{enumerate}
\item How to theoretically analyze the relationships among the popular image models?
\item Why and how does the combination of complementary models help improve the IR results?
\item What are the effective approaches to jointly exploit multiple models for IR?
\item What types of image models do the deep learning methods inexplicitly exploit?
\end{enumerate}
}
\hl{
To the best of our knowledge, few theory has to date studied and addressed the above questions. 
To investigate the hybrid image recovery methods, it is important to understand whether certain image models are identical, more general, or more restrictive than others.
In this work, we answer the proposed questions, via a 
systematical study of several popular image models, such as sparsity, group-wise sparsity, joint sparsity, and low-rankness, in terms of the \textit{sets} of images that satisfy the corresponding models. 
Such sets will be referred to as the \textit{model sets} for simplicity, and we show the relationships among the model sets. 
We use image denoising, as the simplest IR application, to provide theoretical analysis of model effectiveness for IR, as well as extensive experimental results as the evidence.
Assuming the fact that images satisfy the model sets are the only priors, we denoise the images using the maximum likelihood estimate, by projecting the noisy measurements onto the corresponding model sets.  
We demonstrate how our study can serve as the guidance for boosting the IR results, by combining multiple regularizes based on complementary models, as well as the popular deep models.
}

Our contributions in this paper are summarized as follows:
\begin{itemize}
	\item \hl{ We provide a theoretical analysis on popular image model sets and their relationships (see Section~\ref{sec3} and Fig.~\ref{fig:Relationship}).}

 	\item \hl{We evaluate how effective each image model is for denoising, in terms of the \textit{modeling error} and \textit{survived noise energy} (see Section~\ref{sec4}).}

	
	\item \hl{We relate the denoising performance improvement by combining two image models, to the relationship of their model sets (see Section~\ref{sec43}). }


	\item \hl{Extensive experiments are conducted, comparing the denoising results using single model, and those by combining models with large or small model set intersections.}

	\item \hl{We quantitatively demonstrate the image properties that are inexplicitly exploited by deep learning algorithms for IR. We further improve the denoising results, by combining the state-of-the-art deep learning algorithm with its complementary image models.}
\end{itemize}

The rest of the paper is organized as follows.
Section \ref{sec2} summarizes the related image restoration works based on each of the popular image models.
Section \ref{sec3} provides a theoretical analysis on the popular image models, such as sparsity, joint sparsity, group-wise sparsity, low-rankness, etc. We show the relationships among their solution sets, with mild assumption.
Section \ref{sec4} presents numerical results analyzing how effective certain image model, or the combination of several image models can represent the image, and how robust they are to noise corruption.
Section \ref{sec5} demonstrate the behavior of the proposed image restoration framework using multiple image regularizers. 
We show promising denoising results by combining complementary image priors, which 
also boosts the performance of the state-of-the-art image denoising algorithm based on deep learning.
Section \ref{sec6} concludes with proposal for future works.

\section{Related Works} \label{sec2}

Many recent works focused on model-based image restoration and imaging problems, which are associated with different sensing operators $\Ab$'s in (\ref{eq:ip}).
The regularizers that have been applied in these algorithms, are based on common image models, including sparsity, group-wise sparsity, joint sparsity, low-rankness, etc.
We take the simplest image restoration, namely image denoising, and survey the relevant and representative works according to the image models they applied.
Besides the model-based algorithms, there are other effective and popular image denoising algorithms, such as BM3D \cite{Dabov2007,dabov2007color}, EPLL \cite{zoran2011learning}, etc, using collaborative filtering or probabilistic model.
We restrict our discussion to only image denoising algorithms based on explicit parsimonious models in this paper.
Similar types of image models have also been widely applied in other image restoration applications.

\begin{table}[t!]
\centering
\fontsize{9}{12pt}\selectfont
\begin{tabular}{|c|c|c|c|c|}
\hline
\multirow{2}{*}{\textbf{ Methods }} & \multirow{2}{*}{Sparsity} & Group. & Joint  & Low- \\
 &  & Sparsity & Sparsity & Rankness \\
\hline
DCT \cite{yu2011dct} & \cmark & & & \\
\hline
Wavelets \cite{chang2000adaptive} & \cmark & & & \\
\hline
KSVD \cite{elad} & \cmark & & & \\
\hline
Analysis & \multirow{2}{*}{\cmark} & & & \\
KSVD \cite{rubinstein2013analysis} &  & & & \\
\hline
OCTOBOS &  \cmark & & & \\
\hline
SSC-GSM \cite{dong2015image} & & \cmark & & \\
\hline
GSRC \cite{zha2017image} & & \cmark & & \\
\hline
PGPD \cite{xu2015patch} & & \cmark & & \\
\hline
NCSR \cite{dong2013nonlocally} & & \cmark & & \\
\hline
LSSC \cite{Mairal2009} & & & \cmark & \\
\hline
SAIST \cite{Dong2013} &  & & & \cmark \\
\hline
WNNN  \cite{gu2017weighted} &  & & & \cmark \\
\hline
PCLR \cite{chen2015external} &  & & & \cmark \\
\hline
STROLLR \cite{wen2017sparsity} &  \cmark & & & \cmark\\
\hline
\end{tabular}
\caption{Comparison of the major image models that the popular image denoising methods apply. }
\label{Tab:compareDenoising}
\vspace{-0.2in}
\end{table}

\subsection{Sparsity}
Sparsity of natural signals has been widely exploited for image denoising.
Conventional methods imposed image sparsity by applying analytical transforms, e.g. discrete cosine transform (DCT) \cite{Dabov2007,yu2011dct} and wavelets \cite{chang2000adaptive}.
\hl{Images are approximately sparse in the transform domain, while noise is randomly distributed}.
Thus, applying shrinkage functions, such as hard or soft thresholding in the transform domain can effectively remove noise.
Recent works focus on \textit{synthesis model} for image modeling, in which a dictionary can be learned, and each image patch is approximately represented as a linear combination of a few sparsely selected dictionary atoms \cite{elad,elad2,mairal2009online}.
The popular KSVD methods \cite{elad,elad2} proposed heuristic algorithms for learning the overcomplete dictionary, which is effective in image denoising.
Besides the synthesis model, other works, including the popular the Analysis KSVD \cite{rubinstein2013analysis} method, proposed dictionary learning algorithm using the \textit{analysis model} \cite{elad2007analysis}.
However, both the analysis and synthesis models involve NP-hard sparse coding step, and expensive learning steps.
As an alternative, very recent methods generalized the analysis model, and proposed the \textit{transform learning} algorithms \cite{sabres3,octobos,wen2017frist} whose sparse coding is exact and cheap.
Structured overcomplete transform learning~\cite{octobos,wen2017frist} was proposed and  demonstrated promising performance in image denoising.

\subsection{Group-wise Sparsity and Joint Sparsity}

Besides sparsity, natural images are known to have self-similarity.
Non-local but similar structures within an image can be grouped and jointly processed, to help restore the image more effectively.
\hl{Recent image denoising algorithms, such as SSC-GSM \cite{dong2015image} and PGPD \cite{xu2015patch}, proposed to exploit such property by applying the \textit{group-wise sparsity} model, in which similar image patches are first grouped, and a different dictionary is learned within each group for IR.
Such approaches demonstrated advantages for recovering images with diverse textures \cite{zha2017image,dong2013nonlocally,dong2015image,xu2015patch}.
As an alternative, Mairal et al. proposed the LSSC method \cite{Mairal2009} which constrained the sparse codes within each group of similar patches to be not only sparse, and also share the same support of their sparse codes.
Such image model is called \textit{joint sparsity} \cite{Mairal2009}, which is more restrictive for imposing the intra-group data correlation.}

\subsection{Low-Rankness}

Another popular approach to exploit image non-local self-similarity, is to impose \textit{low-rankness} of groups of similar patches.
A successful approach of this nature vectorizes the image patches, to form the columns of a data matrix for each group.
Such data matrix is restored by low-rank approximation, and its columns are then aggregated to recover the image~\cite{gu2017weighted}.
Image denoising algorithms, including WNNM \cite{gu2017weighted}, SAIST \cite{Dong2013}, PCLR \cite{chen2015external}, based on low-rank image prior have demonstrated superior performance in image recovery applications.
Recently proposed STROLLR \cite{wen2017sparsity} further improves the quality of the denoised estimate by simultaneously applying low-rankness and sparsity models.

\subsection{Bridging the Gap Between Models}

There are a handful of previous efforts on bridging the gap between various image models.
\hl{Dong et al. \cite{Dong2013} showed that the joint sparsity model is equivalent to low-rank model in a single-group case.
Such result is limited as the image self-similarity is always exploited by modeling with multiple groups of patches.
Recently, Zha et al. \cite{zha2017analyzing} proposed to construct a specifically designed dictionary for sparse coding. 
It corresponds to a special sparse model, which is showed to be equivalent to the low-rank model.
However, the results in \cite{zha2017analyzing} are hard to be generalized to the commonly used sparse models.}

\section{Image Model Analysis} \label{sec3}

\hl{In this section, we provide an analysis on various image models that are widely used in image restoration applications.
We show the relationship among the various model sets, which is summarized in Fig. \ref{fig:Relationship}.}

\subsection{Synthesis and Transform Models} \label{sec31}

\hl{Sparsity of natural images are exploited under different sparse signal models.
They suggest that a signal $\yb \in \Rmb^{n}$ can be approximately modeled as its sparse feature $\xb \in \Rmb^{m}$ in certain domains.
We define that $\xb$ is $K$-sparse if $\norm{\xb}_0 \leq K$, where $K \ll m$ is called the sparsity level of $\xb$, \ie the number of non-zeros in $\xb$.
The synthesis model \cite{candes2005decoding,elad,Mairal2009,Dong2011} and transform model \cite{octobos,sabres3,frist,chen2017trainable} are the well-known sparse models that are widely used in IR algorithms.
We show that the two sparse models with the same $K$ become equivalent under the \textit{unitary dictionary assumption}.}

\hl{The \textit{synthesis model} represents a signal $\yb \in \Rmb^{n}$ using a synthesis dictionary $\Db \in \Rmb^{n \times m}$ as $\yb \, = \, \Db \, \xb_s + \eb_s$, where $\xb_s \in \Rmb^{m}$ is $K$-sparse, and $\eb_s$ is the dictionary modeling error which is assumed to be small. 
Given the dictionary $\Db$, the synthesis sparse coding problem is formulated as}
\begin{align} \label{eq:synSparse}
\xb_s = \underset{\xb}{\argm}\: \left \| \yb - \Db \, \xb \right \|_{2}^{2}\;\;\;\; s.t.\;\;\norm{\xb}_0 \leq K\, .
\end{align}

\hl{The \textit{transform model}, provides an alternative approach for data representation. 
It models $\yb$ as approximately sparsifiable using a transform $\Wb \in \Rmb^{m \times n}$, \ie $\Wb \, \yb = \xb_a + \eb_a$, where $\xb_a$ is $K$-sparse, and $\eb_a$ is a small transform-domain modeling error. 
Given the transform $\Wb$, the transform model sparse coding problem is formulated as}
\begin{align} \label{eq:tranSparse}
\xb_a = \underset{\xb}{\argm}\: \left \| \Wb \, \yb - \xb \right \|_{2}^{2}\;\;\;\; s.t.\;\;\norm{\xb}_0 \leq K\, .
\end{align}


\begin{figure*}[!t]
\begin{center}
\begin{tabular}{cccc}
\includegraphics[height=1.18in]{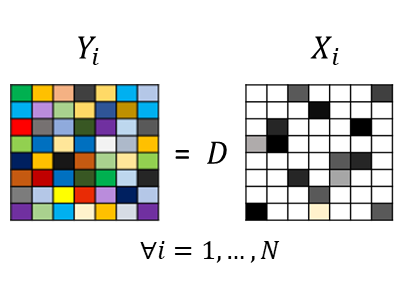} &
\includegraphics[height=1.18in]{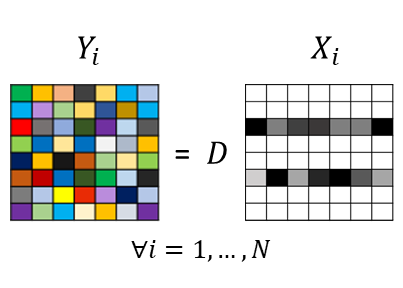} &
\includegraphics[height=1.18in]{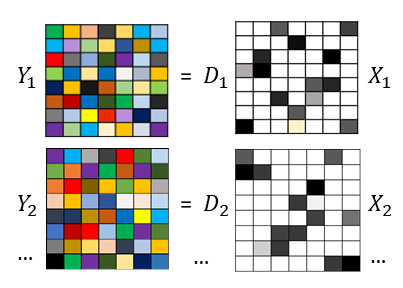} &
\includegraphics[height=1.18in]{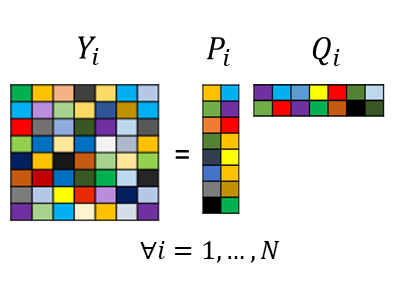} \\
{\small (a) Sparsity (SP)} 	&  {\small (b) Joinst Sparsity (JS)} &  {\small (c) Group-wise Sparsity (GS)} & {\small (d) Low-Rankness (LR) }
\end{tabular}
\vspace{0.1in}
\caption{Illustrations of the signals $\left \{ \mbf{Y}_i \right \}_{i=1}^N$ that satisfy the (a) SP, (b) JS, (c) GS, or (d) LR image models with $K = 2$.}
\label{fig:defModel}
\end{center}
\vspace{-0.1in}
\end{figure*}

\hl{In this work, we unify the two sparse models, by introducing the \textit{unitary dictionary assumption}, \ie   $\Db = \Wb^T \in \Ubb$, and $\Ubb \triangleq \begin{Bmatrix} \Db \in \Rmb^{n \times n} : \Db^{T} \Db = \Ib_n \end{Bmatrix}$ is the set of unitary matrix where $\Ib_n$ is the identity matrix.}
\vspace{0.06in}
\begin{proposition} \label{prop:SP}
The Problems (\ref{eq:synSparse}) and (\ref{eq:tranSparse}) become equivalent, under the \textit{unitary dictionary assumption}.
\end{proposition}
\vspace{0.06in}
\begin{proof}
With the \textit{unitary dictionary assumption}, we have $\Wb \Db = \Ib_n$ and $\left \| \Wb  \Theta  \right \|_2 = \left \| \Theta  \right \|_2  $, $\forall \Theta$. Thus, the objective function in (\ref{eq:synSparse}), \ie $\left \| \yb - \Db \, \xb \right \|_{2}^{2}$ $= \left \| \Wb \yb - \Wb \Db \xb \right \|_{2}^{2}$ $= \left \| \Wb \, \yb - \xb \right \|_{2}^{2}$, becomes identical to that in (\ref{eq:tranSparse}). Therefore, the Problems (\ref{eq:synSparse}) and (\ref{eq:tranSparse}) become equivalent\footnote{Besides, the analysis model~\cite{rubinstein2013analysis,elad2007analysis} also becomes equivalent under the unitary dictionary assumption. We omit the discussion of the analysis model in this work.}.
\end{proof}
\vspace{0.06in}

\hl{In the following analysis, we will only discuss the sparsity using the synthesis model with the unitary dictionary $\Db$ for the sake of simplicity. 
We use the common notations $\xb \in \Rmb^{n}$ and $\eb \in \Rmb^{n}$ to represent the sparse code and modeling error. 
Furthermore, solving synthesis sparse coding problem (\ref{eq:synSparse}) is NP-hard in general~\cite{octobos,AharonEladBruckstein2006}. 
Here, we solve the equivalent problem (\ref{eq:tranSparse}) which has exact solution involving the cheap operator of keeping $K$ elements of signal $\yb$ with largest magnitude, \ie projecting $\yb$ onto the $\ell_0$ ball~\cite{sabres3}. 
The similar equivalence of the synthesis and transform models also holds for joint sparsity, where the $K$-sparse constraint is replaced with the joint $K$-sparse constraint that is defined in Section~\ref{sec32}.}

\subsection{Image Model Definitions} \label{sec32}

To better represent or recover an image $\yb \in \Rmb^{p}$, popular image restoration algorithms investigate the properties of its local patches \cite{Mairal2009,ji2011robust}. 
On top of that, non-local methods group or partition the image patches, via block matching or clustering, before processing them in order to exploit the image self-similarity \cite{Dong2011,dong2013nonlocally,Zhang2014}. 
Following a similar image modeling pipeline, we use a set of data matrices $\begin{Bmatrix} \Yb_i \end{Bmatrix}_{i=1}^{N}$ as the equivalent representation of an image $\yb$. 
Each $\Yb_i \triangleq V_i \, \yb \in \Rmb^{n \times M_i}$ denotes a group patches extracted from $\yb$, \ie a group of vectorized image patches forms the columns of $\Yb_i$.
For simplicity, we use $\Yb_i$ in the following analysis without writing it as a function of $\yb$.  
The grouping operator $V_i$: $\Rmb^{p} \rightarrow \Rmb^{n \times M_i}$ is a function of the image $\yb$ (but this is not displayed explicitly), and its exact form also depends on the specific grouping algorithm.
For a given $\yb$, $V_i$ is treated as a linear operator.
In comparing different models, the $V_i$'s which determine the $\Yb_i$'s, are the same for all models.
Now, it is easily verified that as long as each pixel of $\yb$ appears in at least one $\Yb_i$, the image $\yb$ can be equivalently represented as
\begin{align} \label{eq:imageGroup}
\yb = ( \sum_{i=1}^{N} \, V_i^{\ast} V_i)^{-1} \sum_{i=1}^{N} \, V_i^{\ast} \Yb_i \,\, .
\end{align} 
Here $V_i^\ast :  \Rmb^{n \times M_i} \rightarrow \Rmb^p$ is the adjoint operator of $V_i$: it takes the elements of $\yb$ (the image pixels) found in the input $\Yb_i$, and accumulates them into the output vector $V_i^\ast \Yb_i$. Accordingly, operator $V_i^\ast V_i: \Rmb^p \rightarrow \Rmb^p$ maps an image in $\Rmb^p$ to another such image, and can be represented by a $p \times p$ matrix.

\hl{We now define the various model sets $\Sbb$'s, using the patch block representation $\begin{Bmatrix} \Yb_i \end{Bmatrix}_{i=1}^N$.
We use a superscript to indicate the name of corresponding image model, and the subscript $K$ as the main model parameter. 
We use a superscript to abbreviate the name of corresponding image model, and the subscript $K$ as the main model parameter. 
For example, $\Sbb_{K}^{SP}$ denotes the sparsity model set, with sparsity level $K$.
We assume throughout that $K < \min(n, M_i)$ $\forall i$.
}

The image \textit{sparsity} (SP) model, which was discussed in Section \ref{sec31}, requires each image patch to be sparsifiable under a common unitary dictionary, i.e., each $j$-th column $\Yb_i^j$ of the matrix $\Yb_i$ is sparsifiable (see Fig.~\ref{fig:defModel}(a)).
The image sparsity model set is thus defined as
\begin{definition}[Sparsity] \label{def:SP}
The $K$-sparse set $\Sbb_{K}^{SP} \triangleq \begin{Bmatrix} \yb \in \Rmb^{p} : \exists \, \Db \in \Ubb \, \,s.t.  \Yb_i = \Db \, \Xb_i, \, \norm{\Xb_i^j}_0 \leq K \, \forall i,j \end{Bmatrix}$. An image $\yb$ satisfies the SP model if $\yb \in \Sbb_{K}^{SP}$.
\end{definition}

\hl{On top of sparsity, various works \cite{Mairal2009,ji2011robust} made use of a more restrictive image model - joint sparsity - in order to exploit the correlation of the patches within a group of patches that are similar.
The joint sparsity model \cite{Mairal2009} requires the columns in each $\Xb_i$ to be not only sparse, but also share the same support (see Fig.~\ref{fig:defModel}(b)).
One way to impose joint sparsity of a matrix is by penalizing the $\ell_{0, \infty}$ norm of each $\Xb_i$, which is defined as}
\begin{align} \label{eq:JSnorm}
\left \| \Xb_i \right \|_{0, \infty} \triangleq \sum_{j = 1}^{n} \left \| \Xb_i^{j} \right \|_{\infty}^{0}
\end{align} 
\hl{Here the $\ell_{0, \infty}$ norm simply counts the number of non-zero rows of $\Xb$.
The formal definition of the \textit{joint sparsity} (JS) model set is the following, }
\begin{definition}[Joint Sparsity] \label{def:JS}
The joint $K$-sparse set $\Sbb_{K}^{JS} \triangleq \big\{  \yb \in \Rmb^{p} : \exists \, \Db \in \Ubb \, \,s.t.  \Yb_i = \Db \, \Xb_i, \, \left \| \Xb_i \right \|_{0, \infty} \leq K\; \forall i \big\}$. An image $\yb$ satisfies the JS model if $\yb \in \Sbb_{K}^{JS}$.
%
\end{definition}

Both the SP and JS models apply a common dictionary for all $\left \{ \Yb_i \right \}$.
Recent works \cite{Dong2011,dong2013nonlocally,Zhang2014} relaxed this constraint, applying sparsity by learning a different dictionary $\Db_i$ for each data group $\Yb_i$.
We call this property \textit{group-wise sparsity} (GS) 
\footnote{A similar concept was also named ``group-based'' in previous works. GS is different from the ``joint sparsity'' defined here, which was also sometimes referred to as ``group sparsity'' in other literature~\cite{Mairal2009}.} (see Fig.~\ref{fig:defModel}(c)), 
and the GS model set is defined as follow,
\begin{definition}[Group-wise Sparsity]\label{def:GS}
The group-wise $K$-sparse set $\Sbb_{K}^{GS} \triangleq \big\{ \yb \in \Rmb^{p} : \exists \, \begin{Bmatrix} \Db_i \end{Bmatrix}_{i=1}^{N}, \, \Db_i \in \Ubb \; \,s.t. \,\, \Yb_i = \Db_i \Xb_i, \, \norm{\Xb_i^j}_0 \leq K \,\; \forall i,j \big\}$. An image $\yb$ satisfies the GS model if $\yb \in \Sbb_{K}^{GS}$.
\end{definition}

\hl{One can similarly relax the dictionary sharing constraint on the JS model, and define the group-wise joint sparsity (GJS) model set as following,}
\begin{definition}[Group-wise Joint Sparsity]\label{def:GJS}
The group-wise jointly $K$-sparse set $\Sbb_{K}^{GJS} \triangleq \big\{ \yb \in \Rmb^{p} : \exists \, \begin{Bmatrix} \Db_i \end{Bmatrix}_{i=1}^{N}, \, \Db_i \in \Ubb \; \,s.t. \,\, \Yb_i = \Db_i \Xb_i, \, \norm{\Xb_i}_{0, \infty} \leq K \; \forall i \big\}$. An image $\yb$ satisfies the GJS model if $\yb \in \Sbb_{K}^{GJS}$.
\end{definition}

Besides sparsity related models, low-rankness is another effective prior for exploiting natural image non-local self-similarity \cite{nie2012low,Dong2013,Dong2014,gu2017weighted}.
Most of the image restoration algorithms based on the low-rankness model proposed to group similar image patches, and approximate each data group $\Yb_i$ to be low-rank.
We define the image (group-wise) low-rankness (LR) model set as
\begin{definition}[Low-Rankness] \label{def:LR}
The $K$-rank set $\Sbb_{K}^{LR} \triangleq \big\{ \yb \in \Rmb^{p} : \rank(\Yb_i) \leq K \, \, \forall i \big\}$.
An image $\yb$ satisfies the LR model if $\yb \in \Sbb_{K}^{LR}$.
\end{definition}

\hl{Equivalently, for any $\yb \in \Sbb_{K}^{LR}$, there exists a matrix pair $\mbf{P}_i \in \Rmb^{n \times K}$ and $\mbf{Q}_i \in \Rmb^{K \times M_i}$ for each $\Yb_i$, such that $\Yb_i = \mbf{P}_i  \mbf{Q}_i$. 
We use this interpretation to illustrate the condition of the $K$-rank set in Fig.~\ref{fig:defModel}(d).}
Besides the four popular models, very recent works proposed to exploit the SP and LR properties simultaneously on image and video data, demonstrating superior performance in restoration and reconstruction tasks \cite{wen2017sparsity,wen2017joint}.
We refer to such image models that require image data to be both sparse and low-rank, as the \textit{SPLR model}.
The SPLR image model set is defined as
\begin{definition}[SPLR] \label{def:SPLR}
The joint $K$-sparse and $K$-rank set $\Sbb_{K}^{SPLR} \triangleq \Sbb_{K}^{LR} \cap \Sbb_{K}^{SP}$.
An image $\yb$ satisfies the SPLR model if $\yb \in \Sbb_{K}^{SPLR}$.
\end{definition}

\subsection{Main Results} \label{sec33}

We analyze the relationship among the various sparsity and low-rankness related image models.
The results are presented in terms of the corresponding model sets.
We first consider the special case: a single ($N=1$) group of patches. 

\vspace{0.1in}
\begin{theorem}\label{thm:mainK1}
When $N=1$, the image model sets satisfy
\begin{enumerate}
\item $\Sbb_K^{JS} = \, \Sbb_K^{GJS} = \Sbb_K^{LR} = \Sbb_K^{SPLR}$,
\vspace{0.06in}
\item $\Sbb_K^{SP} = \, \Sbb_K^{GS}$,
\vspace{0.06in}
\item $\Sbb_K^{JS} \subsetneq \Sbb_K^{SP}$.
\vspace{0.06in}
\end{enumerate}
\end{theorem}

\begin{figure}[!t]
\begin{center}
\begin{tabular}{c}
\hspace{-0.1in}
\vspace{-0.1in}
\includegraphics[height=1.8in]{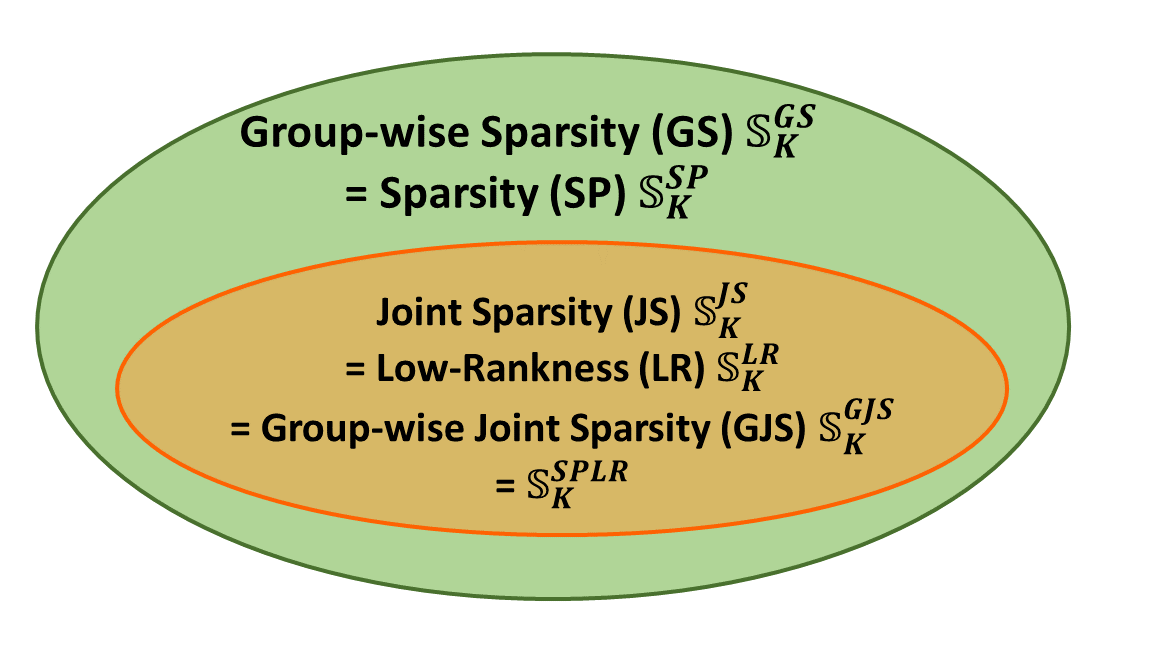}\\
\end{tabular}
\vspace{0.05in}
\caption{A Venn diagram illustrating the relationship among the sets of GS, SP, LR, GJS, JS, and SPLR image models, when $N = 1$.}
\label{fig:RelationshipN1}
\end{center}
\vspace{-0.2in}
\end{figure}

The relationship when $N=1$ is summarized as the Venn diagram in Fig.~\ref{fig:RelationshipN1}. 
Next, we generalize to $N>1$.

\vspace{0.1in}
\begin{theorem}\label{thm:main}
When $N>1$ with $\{ V_i \}_{i=1}^N$ common to all models, the various image model sets satisfy
\vspace{0.06in}
\begin{enumerate}
\item $\Sbb_K^{JS} \subsetneq \Sbb_{K}^{SP} \subsetneq \Sbb_{K}^{GS}$.
\vspace{0.06in}

\item $\Sbb_K^{JS} \subsetneq \Sbb_{K}^{GJS} \subsetneq \Sbb_{K}^{GS}$.

\vspace{0.06in}
\item $\Sbb_{K}^{GJS} = \Sbb_{K}^{LR}$.

\vspace{0.06in}
\item $\Sbb_{K}^{SP} \nsubseteq \Sbb_{K}^{LR}$.
\vspace{0.06in}
\item $\Sbb_{K}^{LR} \nsubseteq \Sbb_{K}^{SP}$ when $N > {n \choose K}$.
\vspace{0.06in}
\item $\Sbb_{K}^{JS} \subsetneq \Sbb_{K}^{SPLR}$.
\end{enumerate}
\end{theorem}
\vspace{0.1in}

Statements (1) and (2) in Theorem \ref{thm:main} are relatively straightforward, indicating relationships of different types of sparsity models.
Statement (3) indicates that the properties of group joint sparsity and low rankness coincide.
Statement (4) and (5) indicate that the SP and LR model sets overlap at most partially (with sufficiently large $N$). 
Because their intersection is the non-empty set $S_K^{SPLR}$, it follows that they do overlap, but only partially. 
Statement (6) indicates that JS is a proper subset of SPLR.
Figure \ref{fig:Relationship} illustrates the main results.

\subsection{Proof of Image Model Set Relationships} \label{sec34}

To prove Theorem~\ref{thm:mainK1} and Theorem~\ref{thm:main} of Section~\ref{sec33}, we first prove
 several Lemmas.
The first two lemmas which hold for any $K$ without additional assumption will be used for proving both Theorem~\ref{thm:mainK1} and Theorem~\ref{thm:main}.

\vspace{0.06in}
\begin{lemma}\label{lem:LReqGJS}
$\Sbb_{K}^{GJS} = \Sbb_{K}^{LR}$.
\end{lemma}
\vspace{0.06in}
\begin{proof}[Proof of Lemma \ref{lem:LReqGJS}]
For any $\yb \in \Sbb_{K}^{LR}$, each $\Yb_i$ in the equivalent representation $\left \{ \Yb_i \right \}_{i=1}^N$ has full SVD as $\Yb_i = \Ub_i \Sigma_i \mbf{Q}_i^{T}$ where
$\Ub_i \in R^{n \times n}$,  $Q_i \in R^{M_i \times M_i}$, and $\Sigma_i \in R^{n \times M_i}$ with main diagonal $\diag(\Sigma_i)_j = 0 \; \forall j > K$.
Let $\Db_i = \Ub_i$, $\Yb_i = \Db_i \Xb_i$ where $\Xb_i = \Sigma_i \mbf{Q}_i^{T}$, and $\left \| \Xb_i \right \|_{0, \infty} \leq \left \| \Sigma_i \right \|_{0, \infty} \leq K$ $\forall i$. Thus, $\yb \in \Sbb_{K}^{GJS}$, which shows $\Sbb_{K}^{GJS} \subseteq \Sbb_{K}^{LR}$.

On the other hand, for any $\yb \in \Sbb_{K}^{GJS}$, each $\Yb_i = \Db_i \Xb_i$ such that $\left \| \Xb_i \right \|_{0, \infty} \leq K$. 
Thus $\rank( \Xb_i ) \leq K$, and $\rank( \Yb_i ) \leq \rank( \Xb_i ) \leq K$.
Therefore, $\yb \in \Sbb_{K}^{LR}$ and $\Sbb_{K}^{LR} \subseteq \Sbb_{K}^{GJS}$, which shows $\Sbb_{K}^{GJS} = \Sbb_{K}^{LR}$.
\end{proof}
\vspace{0.06in}

\vspace{0.06in}
\begin{lemma}\label{lem:JSsubSP}
$\Sbb_{K}^{JS} \subsetneq \Sbb_{K}^{SP}$.
\end{lemma}
\vspace{0.06in}
\begin{proof}[Proof of Lemma \ref{lem:JSsubSP}]
For any $\yb \in \Sbb_{K}^{JS}$, each $\Yb_i = \Db \, \Xb_i$ with $\left \| \Xb_i \right \|_{0, \infty} \leq K$ which counts the number of non-zero rows of $\Xb_i$. Since $\left \| \Xb_i^j \right \|_0 \leq \left \| \Xb_i \right \|_{0, \infty} \leq K$ $\forall i,\,j$, we have $\yb \in \Sbb_{K}^{SP}$ .

On the other hand, for $\yb \in \Sbb_{K}^{SP}$, \ie $\left \| \Xb_i^j \right \|_0 \leq K$. The condition $\left \| \Xb_i^j \right \|_0 = \left \| \Xb_i \right \|_{0, \infty} \leq K$ $\forall i,\,j$ holds, only if all $\Xb_i^j$ share the same supports. Otherwise, $\yb \notin \Sbb_{K}^{JS}$.
\end{proof}
\vspace{0.16in}

We now prove the Lemma~\ref{lem:SPeq} and Lemma~\ref{lem:JSeq} which are relatively trivial, and hold only when $N = 1$.

\vspace{0.06in}
\begin{lemma}\label{lem:SPeq}
$\Sbb_K^{SP} = \Sbb_K^{GS}$, when $N=1$.
\end{lemma}
\vspace{0.06in}
\begin{proof}[Proof of Lemma \ref{lem:SPeq}]
Since $N=1$, there is only one group in the representation, \ie $\Yb_1$ is the equivalent representation of  $\yb$.
Thus the shared dictionary $\Db$ in $\Sbb_K^{SP}$ is equivalent to the $\Db_1$ in $\Sbb_K^{GS}$.
Therefore, $\Sbb_K^{SP} = \Sbb_K^{GS}$.
\end{proof}
\vspace{0.06in}

\vspace{0.06in}
\begin{lemma}\label{lem:JSeq}
$\Sbb_K^{JS} = \, \Sbb_K^{GJS} = \, \Sbb_K^{LR} = \, \Sbb_K^{SPLR}$, when $N=1$.
\end{lemma}
\vspace{0.06in}
\begin{proof}[Proof of Lemma \ref{lem:JSeq}]
Since $N=1$, similar to the proof of Lemma~\ref{lem:SPeq}, $\Sbb_K^{JS} = \, \Sbb_K^{GJS}$.
By Lemma~\ref{lem:LReqGJS}, $\Sbb_K^{JS} = \, \Sbb_K^{GJS} = \, \Sbb_K^{LR}$. 
Finally, by Lemma~\ref{lem:JSsubSP}, the intersection $\Sbb_K^{SPLR} = \Sbb_K^{LR}$, which completes the proof. 
\end{proof}
\vspace{0.06in}

Lemmas~\ref{lem:LReqGJS} to \ref{lem:JSeq} together prove Theorem~\ref{thm:mainK1}, which states the relationship of the model sets when $N = 1$.
We now consider the general case when $N > 1$, and show the following lemmas to prove Theorem~\ref{thm:main}.
We first show Lemma~\ref{lem:JSsubSPsubGS} and Lemma~\ref{lem:JSsubGJSsubGS}, which are relatively trivial.


\vspace{0.06in}
\begin{lemma}\label{lem:JSsubSPsubGS}
$\Sbb_K^{JS} \subsetneq \Sbb_{K}^{SP} \subsetneq \Sbb_{K}^{GS}$, when $N>1$.
\end{lemma}
\vspace{0.06in}
\begin{proof}[Proof of Lemma \ref{lem:JSsubSPsubGS}]
For any $\yb \in \Sbb_{K}^{SP}$, $\Yb_i = \Db \, \Xb_i$ with $\left \| \Xb_i^j \right \|_0 \leq K$ $\forall i, \, j$. Let $\Db_i = \Db$ $\forall i$, thus we have $\Yb_i = \Db_i \, \Xb_i$ with $\left \| \Xb_i^j \right \|_0 \leq K$ $\forall i,\, j$. Therefore, $\yb \in \Sbb_{K}^{GS}$. On the other hand, there exists $\yb \in \Sbb_{K}^{GS}$ but $\yb \notin \Sbb_{K}^{SP}$. We first consider $N=2$ and construct a counter example $\Yb_1 = \Db_1 \Xb_1$ and $\Yb_2 = \Db_2 \Xb_2$ where $\left \| \Xb_i^j \right \|_0 \leq K$ $\forall j$ and $i = 1,2$, and  
\begin{equation} 
\nonumber 
\Db_1 = \mbf{I}_n
, \,
\Db_2 = 
\begin{bmatrix}
\cos(\theta) & -\sin(\theta) & \mbf{0}  \\ 
\sin(\theta) & \cos(\theta) & \vdots \\ 
\mbf{0} & \dotsc & \mbf{I}_{n-2} \\ 
\end{bmatrix} \, ,
\end{equation}
with any $\theta \neq 2 l \pi$, $l \in \mathbb{Z}$, and identity matrix $\mbf{I}_n \in \Rmb^{n \times n}$.
Since $\Db_1 = \mbf{I}_n$, $\left \| \Yb_1^j \right \|_0 \leq K$ $\forall j$
Furthermore, Lemma~\ref{lem:JSsubSP} shows that $\Sbb_K^{JS} \subsetneq \Sbb_{K}^{SP}$, which completes the proof.
\end{proof}
\vspace{0.06in}


\vspace{0.06in}
\begin{lemma}\label{lem:JSsubGJSsubGS}
$\Sbb_K^{JS} \subseteq \Sbb_{K}^{GJS} \subseteq \Sbb_{K}^{GS}$, when $N>1$.
\end{lemma}
\vspace{0.06in}
\begin{proof}[Proof of Lemma \ref{lem:JSsubGJSsubGS}]
For any $\yb \in \Sbb_{K}^{JS}$, each $\Yb_i = \Db \, \Xb_i$ with $\left \| \Xb_i \right \|_{0, \infty} \leq K$. Let $\Db_i = \Db$ $\forall i$, then $\yb \in \Sbb_{K}^{GJS}$. Thus $\Sbb_K^{JS} \subseteq \Sbb_{K}^{GJS}$
Furthermore, for any $\yb \in \Sbb_{K}^{GJS}$, $\Yb_i = \Db_i \, \Xb_i$ with $\left \| \Xb_i \right \|_{0, \infty} \leq K$ $\forall i$.
Since $\left \| \Xb_i^j \right \|_0 \leq \left \| \Xb_i \right \|_{0, \infty} \leq K$ $\forall i,\,j$, we have $\yb \in \Sbb_{K}^{GS}$ and thus $\yb \in \Sbb_{K}^{GS}$, which completes the proof.
\end{proof}
\vspace{0.06in}

Lemmas~\ref{lem:JSsubSPsubGS} and \ref{lem:JSsubGJSsubGS} show that the GJS and SP sets are both the supersets of JS, and also both are subsets of GS.
We now show that neither LR nor SP include one another by Lemmas~\ref{lem:SPnsubLR} and \ref{lem:LRnsubSP}.

\vspace{0.06in}
\begin{lemma} \label{lem:SPnsubLR}
$\Sbb_{K}^{SP} \nsubseteq \Sbb_{K}^{LR}$.
\end{lemma}
\vspace{0.06in}
\begin{proof}[Proof of Lemma \ref{lem:SPnsubLR}]
For $\yb \in \Sbb_{K}^{SP}$, $\exists \Db \in \Ubb$ such that $\Yb_i = \Db \, \Xb_i$ with $\norm{\Xb_i^j}_0 \leq K$  $\forall i,j$, and $\rank(\Yb_i) = \rank(\Xb_i)$.
Whereas, $\rank(\Xb_i)$ may not be smaller than $K$. 
As a counter example, we can construct $\Xb_i$ as
\begin{equation} 
\nonumber
\Xb_i =  \begin{bmatrix}
    \Lambda_K & \mbf{B}_i^U\\
    \mbf{0} & \mbf{B}_i^L
  \end{bmatrix}, \; 
\Lambda_K =  \begin{bmatrix}
    0 & 1 & \dotsc & 1 \\
    1 & 0 & 1 & 1 \\
		\vdots & 1 & \ddots & 1 \\
		1 & \dotsc & 1 & 0 
  \end{bmatrix} , \,
\mbf{B}_i \triangleq \begin{bmatrix}
    \mbf{B}_i^U \\
    \mbf{B}_i^L
  \end{bmatrix}
\end{equation}
Here the circular matrix $\Lambda_K \in \Rmb^{(K + 1) \times (K + 1)}$ is full-rank.
If the size of $\Xb_i$ is larger than $\Lambda$, we pad zero rows below $\Lambda$ when $n > K + 1$, and random $\mbf{B}_i$ with $\norm{\mbf{B}_i^j}_0 \leq K$ $\forall j$ when $M_i > K + 1$.
The $\Xb_i$ satisfies $\norm{\Xb_i^j}_0 \leq K$ $\forall j$ but $\rank(\Xb_i) = K + 1$. Thus $\yb \notin \Sbb_{K}^{LR}$, which completes the proof.
\end{proof}
\vspace{0.06in}

\vspace{0.06in}
\begin{lemma} \label{lem:LRnsubSP}
$\Sbb_{K}^{LR} \nsubseteq \Sbb_{K}^{SP}$, when $N > {n \choose K}$.
\end{lemma}
\vspace{0.06in}
\begin{proof}[Proof of Lemma \ref{lem:LRnsubSP}]
For an image signal $\yb \in \Sbb_{K}^{LR}$, i.e., $\left \{ \Yb_i \right \}_{i=1}^N$ with each $\Yb_i \in \Rmb^{n \times M}$, we first consider $M = K$. 
Thus we have $NK$ patch vectors, denoted as $\Yb_i^j \in \Rmb^{n}$, e.g.,  $\Yb_i^j$ is the $j$-th column of $\Yb_i$. 
Without loss of generality, in this proof we assume that all $\left \{ \Yb_i^j \right \}$ are in general position, thus any $K+1$ vectors are linearly independent. Assuming the contrary of the Lemma, i.e., $\yb \in \Sbb_{K}^{SP}$, thus $\exists \Db \in \Ubb$ such that $\Yb_i^j = \Db \, \Xb_i^j$ with $\norm{\Xb_i^j}_0 \leq K$  $\forall i, j$. Given any $K$ atoms of $\Db$, there is no more than $K$ vectors from $\left \{ \Yb_i^j  \right \}$ that can be spanned by these $K$ atoms, because any $K+1$ vectors are linearly independent. There are ${n \choose K}$ sets of $K$ atoms in total, thus $\Db$ can at most sparsify ${n \choose K}\,K$ vectors from $\left \{ \Yb_i^j \right \}$. Since ${n \choose K}\,K < NK$, there is a contradiction, which means the contrary of the Lemma is false. Therefore, $\Sbb_{K}^{LR} \nsubseteq \Sbb_{K}^{SP}$, when $N > {n \choose K}$. When $M > K$, one can construct the $\Yb_i$ by simply adding more columns to the $n \times K$ matrices while maintaining the rank $K$, and the result still holds.
\end{proof}

\vspace{0.06in}
\begin{lemma} \label{lem:JSinSPLR}
$\Sbb_{K}^{JS} \subsetneq \Sbb_{K}^{SPLR}$, when $N>1$.
\end{lemma}
\vspace{0.06in}
\begin{proof}[Proof of Lemma \ref{lem:JSinSPLR}]
Based on Lemma~\ref{lem:LReqGJS}, \ref{lem:JSsubSPsubGS}, and \ref{lem:JSsubGJSsubGS}, $\Sbb_{K}^{JS} \subseteq \Sbb_{K}^{SP}$ and $\Sbb_{K}^{JS} \subseteq \Sbb_{K}^{LR}$. Therefore, $\Sbb_{K}^{JS} \subseteq \Sbb_{K}^{SPLR}$.

We now only need to show $\Sbb_{K}^{JS} \neq \Sbb_{K}^{SPLR}$.
For $\yb \in \Sbb_{K}^{SPLR}$, $\exists \Db \in \Ubb$ such that $\Yb_i = \Db \, \Xb_i$ with $\norm{\Xb_i^j}_0 \leq K \, \forall i, j$. Furthermore, $\rank(\Yb_i) \leq K$, which means $\rank(\Xb_i) \leq K$ $\forall i$.
Such $\Xb_i$ may not satisfy the joint sparsity condition, \ie $\left \| \Xb_i \right \|_{0, \infty} \leq K$.
As a counter example for $N=2$, we can construct $\Xb_1$ and $\Xb_2$ as
\begin{equation} 
\nonumber
\Xb_1 =  \begin{bmatrix}
 1 & \dotsc & 1 & 1 & 0 \\ 
 1 &  & 1 & 0 & K-1 \\ 
 \vdots &  & \text{\reflectbox{$\ddots$}} & 1 & -1\\ 
 1 & 0 &  & \vdots & \vdots \\ 
 0 & 1 & \dots & 1 & -1 
\end{bmatrix} \in \Rmb^{(K + 1) \times (K + 1)}
\end{equation}
\begin{equation} 
\nonumber  
\Xb_2 =  \begin{bmatrix}
 -1 & 1 & \dotsc & 1 & 0 \\ 
 \vdots &  &  & 0 & 1 \\ 
 -1 & 1 & \text{\reflectbox{$\ddots$}} &  & \vdots \\ 
 K-1 & 0 & 1 &  & 1\\ 
 0 & 1 & 1 & \dotsc & 1
\end{bmatrix}  \in \Rmb^{(K + 1) \times (K + 1)}
\end{equation}
If the size of $\Xb_i$ is larger than $(K + 1) \times (K + 1)$, we pad zero rows below $\Xb_i$ when $n > K + 1$, and repeat any column of $\Xb_i$ if $M_i > K + 1$.
We have $\norm{\Xb_1^j}_0 \leq K$ and $\norm{\Xb_2^l}_0 \leq K$ $\forall j,l$. Furthermore, $\rank(\Xb_1) = \rank(\Xb_2) = K$ because one of the columns satisfy the following
\begin{align} 
\nonumber \Xb_1^{K+1} = & \sum_{j=1}^{K-1} \Xb_1^{j} - (K-1) \Xb_1^{K} \\
\nonumber \Xb_2^1 = & \sum_{j=3}^{K+1} \Xb_2^{j} - (K-1) \Xb_1^{2}
\end{align}
However, the joint sparsity $\left \| \Xb_1 \right \|_{0, \infty} = K + 1$, thus $\yb \notin \Sbb_{K}^{JS}$. 
Therefore, $\Sbb_{K}^{JS} \neq \Sbb_{K}^{SPLR}$, which completes the proof.
\end{proof}
\vspace{0.06in}

The Lemma~\ref{lem:JSinSPLR} shows the Statement (6) in the Theorem~\ref{thm:main}. Therefore, we complete the proof of the Theorem \ref{thm:main}.

\section{Image Modeling and Denoising} \label{sec4}

Since natural images are neither exactly sparse nor exactly low-rank, the commonly used image models are all the approximate models, i.e., the true image data are close to, but not exactly belong to the image model sets.
\hl{Therefore, on top of the analysis of the relationship among image model sets, 
we study how effective each model can be applied to represent, and thus to denoise image data.
For image denoising, an effective image model should be able to}
\begin{enumerate}
\item Preserve the clean image, \ie the model set is close to the distribution of natural images.
\item Reject random noise, \ie the model set is small and cannot be too flexible.
\end{enumerate}
\hl{We propose to study the image denoising, which is the simplest restoration problem, in order to quantitatively evaluate the effectiveness of image models.
Note that effective image models in denoising problems are usually also useful in other restoration or inverse problems \cite{dong2013nonlocally,Zhang2014,wen2017frist}.}

\subsection{Denoising by Projection} \label{sec41}

Denote the noisy measurement of a clean signal $\mbf{u}$ as $\mbf{z} = \mbf{u} + \mbf{e}$, where $\mbf{e}$ is additive white Gaussian noise.
Assuming the fact that $\mbf{u}$ satisfies a certain model, \ie belongs to a certain model set is the only prior, we denoise $\mbf{z}$ using the maximum likelihood estimate of $\mbf{u}$, by projecting $\mbf{z}$ onto the corresponding model set.

Though each of the discussed image model sets corresponds to a union of subspaces, locally, image patch denoising can be viewed as projection onto a low-dimensional subspace, \eg 
sparse coding with a specific support corresponds to projection onto the subspace spanned by the selected atoms.
Figure~\ref{fig:subspaceApprox} provides a simple illustration of denoising using the SP model with $n=2$ and $K=1$.
Thus, we approximate the denoised estimate of $\mbf{z}$ as $f(\mbf{z}) = \mbb{P} \mbf{z}$, where the operator $\mbb{P}$ denotes the projection onto the local subspace of certain model set.

To simplify the analysis of model effectiveness in denoising, we first investigate the denoising of image patches. 
Unlike complete image denoising, this simplified approach does not involve patch consensus or aggregations.
We work with a set of image patches from an image corpus $\Omega$, which are denoted as $\begin{Bmatrix} \mbf{u}_i \end{Bmatrix}_{i \in \Omega}$.
The noisy measurement of each $\mbf{u}_i$ is $\mbf{z}_i = \mbf{u}_i + \mbf{e}_i$, where $\mbf{e}_i$ is additive noise.
The proposed denoising schemes are consistent with the definitions of various models in Section~\ref{sec32}.

\begin{figure}[!t]
\begin{center}
\begin{tabular}{c}
\includegraphics[width=2.4in]{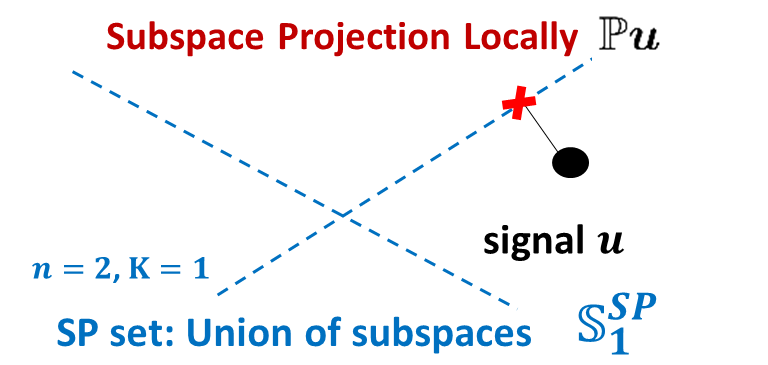}\\
\end{tabular}
\vspace{0.05in}
\caption{Illustration of denoising signal in 2D using SP model with $K=1$.}
\label{fig:subspaceApprox}
\end{center}
\vspace{-0.18in}
\end{figure}

\subsection{Denoising with a Single Model} \label{sec42}

We first describe and analyze the denoising methods for image patches by applying a single image model, including the SP, GS, JS, and LR models.

\subsubsection{Sparsity Model (SP)}
The denoising algorithm based on a SP model projects each $\mbf{z}_i \in \Rmb^{n}$ onto a union of subspaces spanned by $K$ atoms of the underlying dictionary, where $K$ is the patch sparsity level.
The $n \times n$ unitary dictionary $\mbf{D} \in \Rmb^{n}$ can be learned by solving the following problem,
\begin{align} \label{eq:spDL}
\nonumber \underset{\{ \Db, \left \{ \mbf{x}_i \right \} \}}{\operatorname{min}} \sum_{i \in \Omega} & \left \| \mbf{z}_i - \Db \mbf{x}_i  \right \|_{2}^{2} \\
s.t. & \, \left \| \mbf{x}_i \right \|_0 \leq K\,\, \forall i\,,\;\; \Db^T \Db = \mbf{I}_n.
\end{align}
Here $\mbf{x}_i \in \Rmb^{n}$ denotes the sparse code for $\mbf{z}_i$, which has at most $K$ non-zero elements. 
As proved in Section \ref{sec31}, unitary dictionary learning problem is equivalent to unitary transform learning \cite{sabres3,wen2017sparsity}. 
A simple block coordinate descent algorithm can be applied to solve for $\Db$ and $\begin{Bmatrix} \mbf{x}_i \end{Bmatrix}$ iteratively \cite{sabres3,wen2017sparsity}.
Fixing $\Db$, the optimal sparse code $\hat{\mbf{x}}_i = H_K (\Db^T \, \mbf{z}_i)$, where $H_K(\cdot)$ is the projector onto the $K$-$\ell_0$ ball \cite{sabres3}, i.e., $H_K(\mbf{b})$ zeros out all but the $K$ elements with the largest magnitude in $\mbf{b} \in \Rmb^{n}$.
Fixing $\begin{Bmatrix} \mbf{x}_i \end{Bmatrix}$, there is an exact solution for the optimal dictionary $\hat{\Db} = \mbf{G}\, \mbf{S}^T$, where the square matrices $\mbf{G}$ and $\mbf{S}$ are obtained using singular value decomposition (SVD) as $\mbf{S}\, \Sigma\, \mbf{G}^T = \mathrm{SVD}(\sum_{i \in \Omega}\, \mbf{z}_i \mbf{x}_i^T)$ \cite{sabres3}.

Given the dictionary $\Dbh$ and the sparsity level $K$, the denoised estimate of each patch is obtained by
\begin{equation} \label{eq:spDenoise}
f^{SP} ( \mbf{z}_i ) = \mbb{P}_i^{SP} \mbf{z}_i = \hat{\mbf{D}} H_K (\hat{\Db}^T \, \mbf{z}_i) \, ,
\end{equation}
where $\mbb{P}_i^{SP}$ denotes the linear projection operator for denoising the $i$-th patch, based on the SP model. 
\footnote{The projection operator is a function of $\Dbh$ and $K$. Though all patches share the common dictionary $\Dbh$, the projection operator varies for each patch as the support of each $\mbf{x}_i$ is different.}

\subsubsection{Group-wise Sparsity (GS)}
Different from sparsity, the denoising algorithm based on GS model is a non-local method.
Similar patches are first matched into $N$ groups, and vectorized to form columns of data matrices $\begin{Bmatrix} \mbf{Z}_i \end{Bmatrix}_{i = 1}^{N}$, where each $\mbf{Z}_i \in \Rmb^{n \times M_i}$.
The GS based algorithm learns separate dictionary $\mbf{D}_i$ for each group, by solving the following problem
\begin{align} \label{eq:gsDL}
\nonumber \underset{\{ \Db_i, \mbf{X}_i \}}{\operatorname{min}} & \left \| \mbf{Z}_i - \Db_i \mbf{X}_i  \right \|_{F}^{2} \;\; \forall i = 1, ..., N\\
s.t. & \, \left \| \mbf{X}_{i, j} \right \|_0 \leq K\,\, \forall j\,,\;\; \Db_i^T \Db_i = \mbf{I}_n.
\end{align}
Here $\mbf{X}_{i, j}$ denotes the $j$-th column of sparse code matrix $\mbf{X}_i$.
Very similar to SP based dictionary learning, there is a simple block coordinate descent algorithm solving each $\hat{\Db}_i$ and $\hat{\mbf{X}}_i$, and each step has exact solution.
The difference is that each $\hat{\Db}_i$ is only trained using patches within the $i$-th group.
The $j$-th column of the denoised $i$-th patch is 
\begin{equation} \label{eq:gsDenoise}
f^{GS} ( \mbf{Z}_{i, j} ) = \mbb{P}_{i, j}^{GS}  \mbf{Z}_{i, j} = \Dbh_i H_K (\Dbh_i^T \, \mbf{Z}_{i, j}) \, , 
\end{equation}
where $ \mbb{P}_{i, j}^{GS}$ denotes the projection operator for denoising the $j$-th patch in the $i$-th group, based on the GS model.

\subsubsection{Joint Sparsity (JS)}

To explore patch correlation within each group, the denoising algorithm based on the JS model projects each group of patches onto the same low-dimensional subspace, spanned by $K$ atoms of the common dictionary $\Db$ for all groups.
The JS model dictionary learning problem is formulated as follow,
\begin{align} \label{eq:jsDL}
\nonumber \underset{\{ \Db, \left \{ \mbf{X}_i \right \} \}}{\operatorname{min}} \sum_{i = 1}^N & \left \| \mbf{Z}_i - \Db \mbf{X}_i  \right \|_{F}^{2} \\
s.t. & \, \left \| \mbf{X}_i \right \|_{\infty}^0 \leq K\,\, \forall i\,,\;\; \Db^T \Db = \mbf{I}_n\;.
\end{align}
Similar to the SP model dictionary learning problem, with $\mbf{X}_i$ fixed, the exact solution of each $\Dbh_i$ can be calculated using SVD.
With $\Db_i$ fixed, the sparse coding step also has exact solution.
The optimal sparse code $\Xbh_i = \tilde{H}_K (\Db^T \mbf{Z}_i)$, where the operator $\tilde{H}_K (\cdot)$ sparsifies the matrix with each column has the same support, i.e., $\tilde{H}_K (\mbf{B})$ zeros out all but the $K$ rows of $\mbf{B} \in \Rmb^{n \times M}$, which have the largest $\left \| \mbf{B}^{j} \right \|_F^2$.
The denoised estimate of the $i$-th group is 
\begin{equation} \label{eq:jsDenoise}
f^{JS} ( \mbf{Z}_{i} ) = \mbb{P}_{i}^{JS} = \Dbh \tilde{H}_K (\Dbh^T \, \mbf{Z}_{i}) \, , 
\end{equation}
where $\mbb{P}_{i}^{JS}$ denotes the projection operator for denoising the $j$-th patch in the $i$-th group, based on the JS model.

\subsubsection{Low-Rankness (LR)}

Apart from sparsity, the denoising algorithm based on the LR model projects each group of patches onto the low-dimensional subspace, spanned by its first $K$ eigenvectors.
The denoised estimate based on LR model for each $\mbf{Z}_i$ is obtained by solving the following problem,
\begin{align} \label{eq:lrDenoising}
\nonumber f^{LR} ( \mbf{Z}_i ) =  \underset{ \mbf{L}_i }{\operatorname{argmin}} & \left \| \mbf{Z}_i - \mbf{L}_i  \right \|_{F}^{2} \,\, \forall i\, \\
s.t. & \, \text{rank} ( \mbf{L}_i ) \leq K\;.
\end{align}
There is an exact solution to the low-rank approximation. 
Applying the SVD to each $\mbf{Z}_i$, i.e., $\mbf{P}_i \text{diag}(\mbf{v}_i) \mbf{Q}_i^T = \mathrm{SVD}( \mbf{Z}_i )$, the denoised low-rank estimate is achieved by projecting the eigenvalues onto the $K$-$\ell_0$ ball as following 
\begin{equation} \label{eq:lrEstimate}
f^{LR} ( \mbf{Z}_i ) = \mbb{P}_i^{LR} \mbf{Z}_i = \mbf{P}_i \text{diag} \begin{Bmatrix} H_K(\mbf{v}_i) \end{Bmatrix} \mbf{Q}_i^T \;.
\end{equation}

\subsection{Denoising with Multiple Models} \label{sec43}

We showed that the various image models, which explicitly exploit different image properties, are in fact all related.
Applying multiple image models for IR can potentially provide more effective image data representation.
\hl{Though it is unclear what the best combination of regularizes is theoretically for natural image modeling.
We now show that the \textit{convex combination} of denoising results using single models, is more effective than applying alternating projection. 
We also provide a theoretical analysis of the improvement of denoising results using convex combination. 
}

\begin{figure}[!t]
\begin{center}
\begin{tabular}{c}
\includegraphics[width=2.6in]{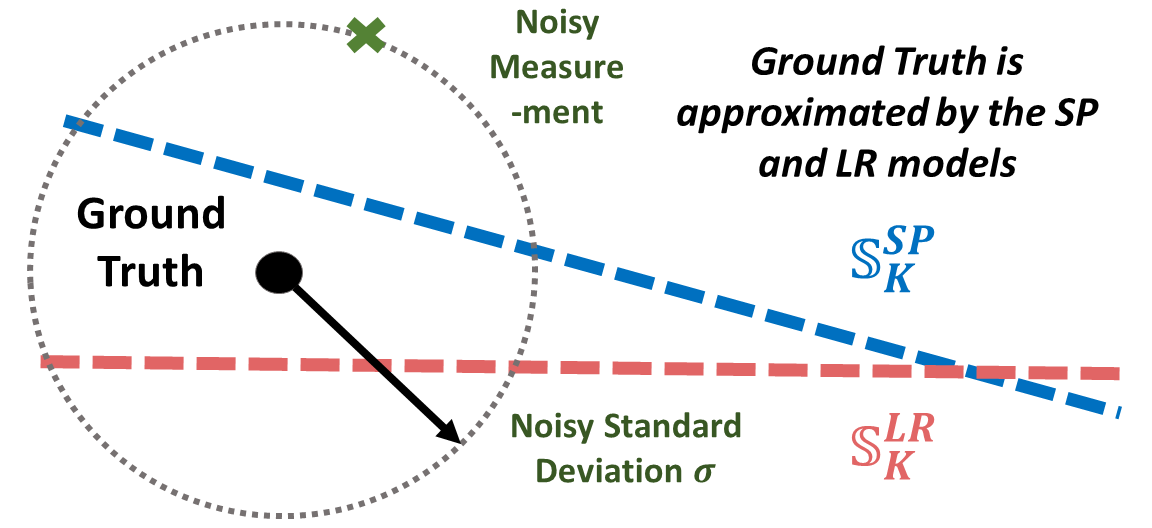}\\
\includegraphics[width=2.6in]{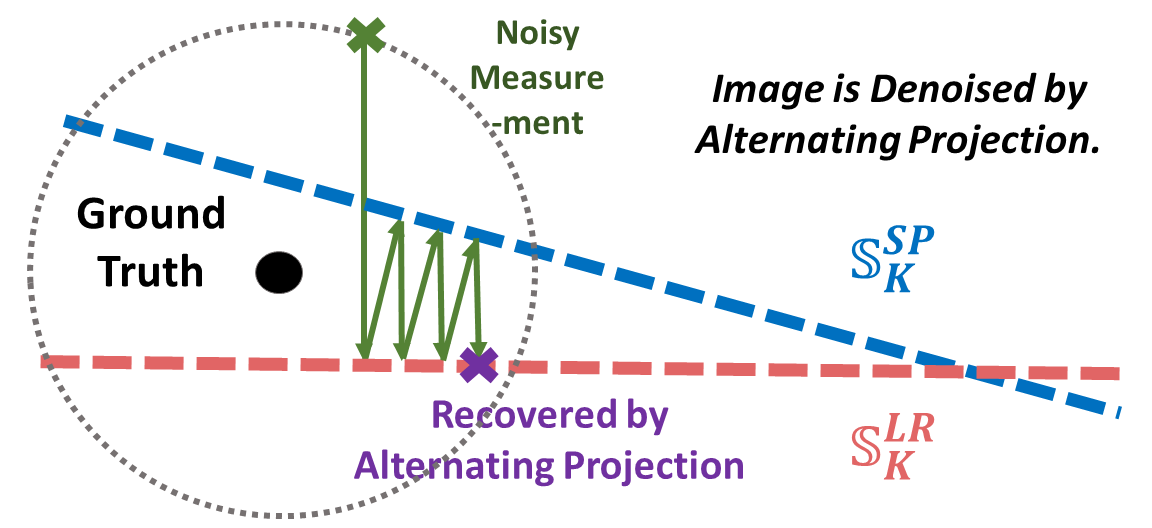}\\
\includegraphics[width=2.6in]{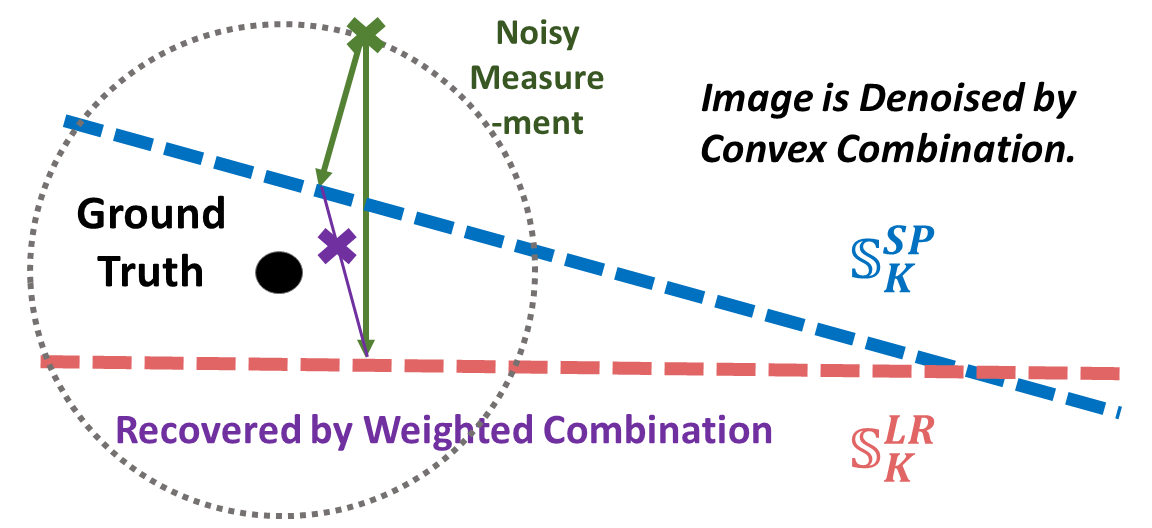}\\
\end{tabular}
\vspace{0.05in}
\caption{Denoising comparison in 2D using dual image models. Top: A noisy measurement is generated from ground true image which is approximated by LR and SP models; Middle: The denoised estimate using alternating projection; Bottom: The denoised estimate using convex combination.}
\label{fig:alterProj}
\end{center}
\vspace{-0.2in}
\end{figure}

\subsubsection{Alternating Projection}
One conventional way to jointly impose multiple image models, is by alternating projection.
Take the case of dual models as example, the method aims to recover the image data by projecting the noisy measurement $\mbf{z}$ onto the set $\mbb{S}^A$ of the model $A$, and the set $\mbb{S}^B$ of the model $B$, iteratively.
Denote the denoised estimate after $t$ times of alternating projection as $f_t^{A + B} ( \mbf{z} )$, which can be represented recursively as
\begin{equation} \label{eq:alterProj}
f_t^{A + B} ( \mbf{z} ) = \mbb{P}^{B} \mbb{P}^{A} f_{t-1}^{A + B} ( \mbf{z} )\;,
\end{equation}
where the initial $f_0^{A + B} ( \mbf{z} ) = \mbf{z}$.
However, since none of the image models can exactly represent natural images, alternating projection is not guaranteed to converge to the ground true signal (Fig.~\ref{fig:alterProj} shows one such example).

\subsubsection{Convex Combination}
Alternatively, we propose to denoise image patches, using a convex combination of the denoised estimates by projecting the noisy measurements onto different individual model sets.
The denoised estimate based on dual models is represented as
\begin{equation} \label{eq:weightedSum}
f^{A + B} ( \mbf{z} ) = \mu\, \mbb{P}^{A} \mbf{z} + (1 - \mu)\, \mbb{P}^{B} \mbf{z}\;,
\end{equation}
where the scalar $\mu$ is the combination weight. 
Figure \ref{fig:alterProj} illustrates a comparison in 2D space, between an denoising example using alternating projection (the middle), and that using convex combination (the bottom).
Neither the LR model, nor the SP model can represent image data exactly, but each of them exploit different properties of natural images. 
Thus, the convex combination of the denoised estimates using different single models, can potentially improve the recovery quality.

\subsubsection{Denoising Analysis by Convex Combination}
To gain some intuition why the proposed approach can improve the denoising performance, we decompose the denoised estimates using algorithm based on model A, and B respectively, as 
\begin{equation} \label{eq:errorDecomp2}
f^{A} (\mbf{z}) = \mbf{u} + \tilde{\mbf{e}}_{A},\;\;  f^{B} (\mbf{z}) = \mbf{u} + \tilde{\mbf{e}}_{B}
\end{equation}
Here $\mbf{u}$ is the ground true signal, and the \textit{remaining noise} in the denoised estimates $f^{A} (\mbf{z})$ and $f^{B} (\mbf{z})$ are denoted as $\tilde{\mbf{e}}_{A}$ and $\tilde{\mbf{e}}_{B}$, respectively.
The dual-model denoised estimate is thus
\begin{equation} \label{eq:weighted}
f^{A+B} (\mbf{z}) = \mbf{u} + \mu \, \tilde{\mbf{e}}_{A} + (1 - \mu) \, \tilde{\mbf{e}}_{B}\, .
\end{equation}
Denote the remaining noise in $f^{A+B} (\mbf{z})$ as $\tilde{\mbf{e}}_{A + B}$, which is the convex combination of $\tilde{\mbf{e}}_{A}$ and $\tilde{\mbf{e}}_{B}$.
Without loss of generality, we define that $\Gamma \triangleq \left \| \tilde{\mbf{e}}_{A} \right \|_2  = \text{min} (\left \| \tilde{\mbf{e}}_{A} \right \|_2, \left \| \tilde{\mbf{e}}_{B} \right \|_2 )$, and $\Gamma + \Delta \triangleq \left \| \tilde{\mbf{e}}_{B} \right \|_2$, with $\Delta \geq 0$.
We would like to achieve the improved denoising result with less remaining noise, i.e.,
\begin{equation} \label{eq:improved}
 \left \| \mu \, \tilde{\mbf{e}}_{A} + (1 - \mu) \, \tilde{\mbf{e}}_{B} \right \|_2 < \Gamma.
\end{equation}
The condition of achieving the denoising improvement, \ie (\ref{eq:improved}) is satisfied, is equivalent to
\begin{equation} \label{eq:target}
\mu^2 \Gamma^2 + (1 - \mu)^2 (\Gamma + \Delta)^2 + 2 \mu (1 - \mu) (\tilde{\mbf{e}}_{A}^T \tilde{\mbf{e}}_{B}) < \Gamma^2\,.
\end{equation}
The condition (\ref{eq:target}) leads to the upper bound of the correlation (i.e., the lower bound of the angle) between the two noise vectors $\tilde{\mbf{e}}_{A}$ and $\tilde{\mbf{e}}_{B}$ as following
\begin{equation} \label{eq:correlation}
\text{cos}\, \theta_{A,B} \triangleq \frac{ (\tilde{\mbf{e}}_{A}^T \tilde{\mbf{e}}_{B}) }{ \left \| \tilde{\mbf{e}}_{A} \right \|_2 \left \| \tilde{\mbf{e}}_{B} \right \|_2 } < 
1 - \frac{[2 + (1 - \mu) \gamma] \gamma}{2 \mu (1 + \gamma)} \,,
\end{equation}
where $\gamma \triangleq \Delta / \Gamma \geq 0$ represents the normalized difference in magnitude of the errors, and $\theta_{A,B}$ is the angle between the two error vectors $\tilde{\mbf{e}}_{A}$ and  $\tilde{\mbf{e}}_{B}$. 
To provide intuition of the performance improvement bound (\ref{eq:correlation}), 
Fig. \ref{fig:combinedError} illustrates $\tilde{\mbf{e}}_{A}$ and  $\tilde{\mbf{e}}_{B}$, and their convex combination $\tilde{\mbf{e}}_{A + B}$ different conditions:
\begin{itemize}
 \item  When $\gamma = 0$, Fig \ref{fig:combinedError} (a) shows that the magnitude of the error $\tilde{\mbf{e}}_{A + B}$ in the combined result always decreases as long as $\text{cos} \, \theta_{A,B} < 1$, \ie $\theta_{A,B} \neq 0$.
 
 \item  When $\gamma$ is large, in order to achieve the denoising improvement (\ref{eq:improved}), Fig \ref{fig:combinedError} (b) shows that the $\text{cos} \theta_{A,B}$ needs to be smaller than the bound  (\ref{eq:correlation}), \ie the angle $\theta_{A,B}$ needs to be sufficiently large. 

 \item When $\gamma$ is large, and the the bound  (\ref{eq:correlation}) is unsatisfied, Fig \ref{fig:combinedError} (c) shows one example that the denoising improvement (\ref{eq:improved}) is not achieved.

\end{itemize}

\begin{figure}[!t]
\begin{center}
\begin{tabular}{ccc}
\includegraphics[height=1in]{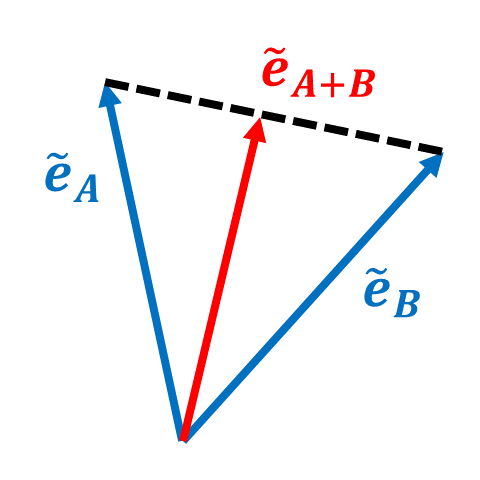} &
\includegraphics[height=1in]{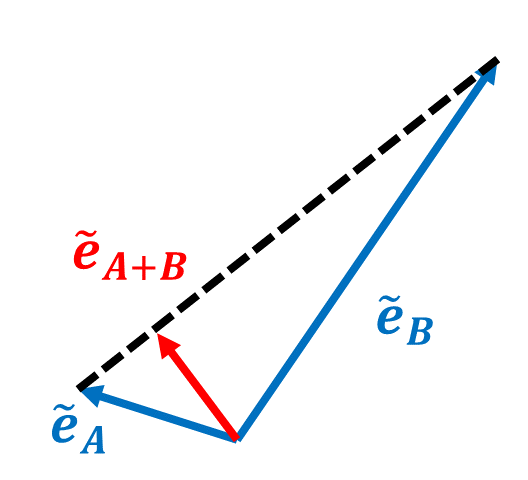} &
\includegraphics[height=1in]{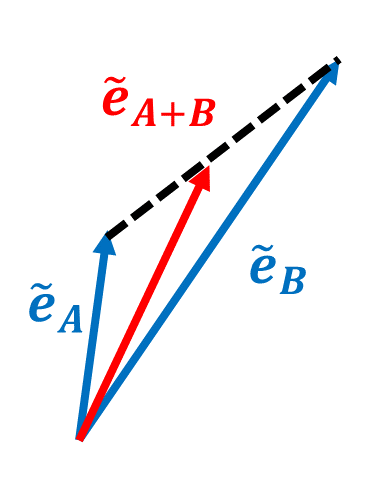}\\
{\small (a)} 	&  {\small (b)} &  {\small (c)}
\end{tabular}
\vspace{0.05in}
\caption{Illustration of the error vectors in the combined estimate $\tilde{\mbf{e}}_{A + B}$, with different conditions of $\tilde{\mbf{e}}_{A}$ and $\tilde{\mbf{e}}_{B}$: (a) $\gamma = 0$, (b) large $\gamma$ with large $\theta_{A,B}$, and (c) large $\gamma$ with small $\theta_{A,B}$. }
\label{fig:combinedError}
\end{center}
\vspace{-0.2in}
\end{figure}



To summarize, we need (1) small $\gamma$ and (2) large $\theta_{A, B}$, to achieve denoising improvement by convex combination.
We now provide suggestions for choosing image denoising algorithms that will satisfy the two corresponding conditions, thus lead to improved denoising results.

\vspace{0.06in}
\begin{proposition}\label{prop:denoising}
Assuming that image denoising using a single model can be approximated as projecting noisy images onto the corresponding subspace, the following suggestions can help boost the denoising performance via convex combination of results using single models $\mbb{S}^A$ and $\mbb{S}^B$. 
\begin{enumerate}
\item Choose denoising algorithms $f^{A} (\cdot)$ and $f^{B} (\cdot)$ with good and similar performance, \ie $\Delta = | \tilde{\mbf{e}}_{A} - \tilde{\mbf{e}}_{B}|$ is small.

\item Among algorithms with similar performances, \ie fixing $\left \| \tilde{\mbf{e}}_{A} \right \|_2$ and $\left \| \tilde{\mbf{e}}_{B} \right \|_2$, select those with small intersection of their model sets, \ie $\mbb{S}^A \cap \mbb{S}^B$ is small.
\end{enumerate}
\end{proposition}
\vspace{0.06in}
\begin{proof}[Proof of Proposition \ref{prop:denoising}]
As we showed in (\ref{eq:correlation}), small $\gamma$ leads to performance improvement.
Since large $\Gamma$ means poor denoising baselines $f^{A} (\cdot)$ and $f^{B} (\cdot)$, the practical option is to obtain small $\Delta$, \ie $f^{A} (\cdot)$ and $f^{B} (\cdot)$ have similar denoising performance.

With the assumption of subspace projection, denote the intersection $\Pi \triangleq \mbb{S}^A \cap \mbb{S}^B$ which is also a subspace.
With $\left \| \tilde{\mbf{e}}_{A} \right \|_2$ and $\left \| \tilde{\mbf{e}}_{B} \right \|_2$ fixed, minimizing $\text{cos}\, \theta_{A,B}$ is equivalent to minimizing $\tilde{\mbf{e}}_{A}^T \tilde{\mbf{e}}_{B}$. 
Thus
\begin{align}
\nonumber \tilde{\mbf{e}}_{A}^T \tilde{\mbf{e}}_{B} = & \, (\mbb{P}_{\Pi} \tilde{\mbf{e}}_{A} + \mbb{P}_{\Pi^\perp} \tilde{\mbf{e}}_{A})^T (\mbb{P}_{\Pi} \tilde{\mbf{e}}_{B} + \mbb{P}_{\Pi^\perp} \tilde{\mbf{e}}_{B}) \\
 = & \, (\mbb{P}_{\Pi} \tilde{\mbf{e}}_{A})^T \mbb{P}_{\Pi} \tilde{\mbf{e}}_{B}
\end{align}
Thus, assuming the remaining noise is uniformly distributed in the model set, smaller (\ie lower-dimensional) $\Pi$ leads to smaller $\tilde{\mbf{e}}_{A}^T \tilde{\mbf{e}}_{B}$, and thus improved denoising performance.

\end{proof}
\vspace{0.06in}

\begin{conj} \label{conj:complement}
When the denoising algorithm is no longer as simple as single subspace projection, the suggestions in Proposition \ref{prop:denoising} still hold for image denoising.
\end{conj}

When the image models become more complicated, it is difficult to provide theoretical analysis on denoising performance.
We, instead, provide experimental results in Section~\ref{sec5}, as the numerical evidences to support our conjectures.

\subsection{Generalization to Image Denoising} \label{sec44}

We generalize the patch denoising method using multiple image models in Section \ref{sec43}, to image-level denoising scheme by combination of multiple algorithms.
Image denoising is typically considered as an inverse problem, which can be formulated as the following optimization problem
\begin{equation} \label{problem:imDenoisingReg} 
\hat{\mbf{x}} = \underset{\mbf{x}}{\operatorname{argmin}} \: \lambda_f \left \| \mbf{x} - \mbf{y} \right \|_{2}^{2} + \Rf\, (\mbf{x})\, ,
\end{equation}
where $\left \| \mbf{x} - \mbf{y} \right \|_{2}^{2}$ is the image fidelity with $\mbf{y}$ being the noisy image, and $\mbf{x}$ being the underlying denoised estimate.
Furthermore, the regularizer $\Rf\, (\mbf{x})$ is imposed based on certain image properties.
There are various image denoising algorithms proposed by exploiting properties based on specific image models.
In order to incorporate models applied in multiple algorithms in one image denoising scheme, we propose a simple image denoising fusion method using the image-level \textit{convex combination}.

Take the dual-model case as an example, the denoised image estimate is obtained by solving the following problem
\begin{align} \label{problem:imDenoisingDualReg} 
\nonumber (\mathrm{P1})\;\;\; \hat{\mbf{x}} =  & \underset{\mbf{x}}{\operatorname{argmin}} \: \;  \lambda_f \left \| \mbf{x} - \mbf{y} \right \|_{2}^{2} \\
\nonumber & + \mu \left \| \mbf{x} - \mbf{x}_A \right \|_2^2 + (1 - \mu) \left \| \mbf{x} - \mbf{x}_B \right \|_2^2 \, \\
\nonumber & = \frac{\lambda_f}{1 + \lambda_f} \mbf{y} + \frac{\mu}{1 + \lambda_f} \mbf{x}_A + \frac{1-\mu}{1+\lambda_f} \mbf{x}_B \, ,
\end{align}
where $\mbf{x}_A \triangleq f^A\,(\mbf{y})$ and $\mbf{x}_B \triangleq f^B\,(\mbf{y})$ are the denoised estimates using the denoising algorithms $f^A\,(\cdot)$ and $f^B\,(\cdot)$, respectively.
If $\lambda_f = 0$, the denoised estimate is simply $\hat{\mbf{x}} = \mu \mbf{x}_A + (1 - \mu) \mbf{x}_B$, which reduces to the case in Section \ref{sec43} when the denoising algorithms are both simple projections.

The denoising scheme ($P1$) can be applied as a computational tool, to evaluate the image models exploited by certain denoising algorithms.
The improvement of the denoising performance by ($P1$) comparing to single algorithm reflects whether the image models are correlated.
Recently, the image denoising algorithms using deep neural networks demonstrated promising performance, while the reason of success remains unclear.
In Section \ref{sec5}, we also apply ($P1$) by combining deep learning methods with various model-based algorithms, and study the image properties that the learned deep neural networks inexplicitly exploit.

\subsection{Evaluation Metrics} \label{sec45}

To quantitatively compare the effectiveness of different image models for denoising, we propose several metrics for evaluating the quality of their denoised estimates.

Suppose we denoise $\mbf{z} = \mbf{u} + \mbf{e}$ using a specific denoiser $f(\cdot)$ via projection, \ie $f(\mbf{z}) = \mbb{P} \mbf{z}$.
The denoised estimate can be decomposed into two parts, namely the \textit{clean signal approximation} $\tilde{\mbf{u}}$ and the \textit{survived noise} $\tilde{\mbf{e}}$ as following
\begin{equation} \label{eq:errorDecomp1}
f(\mbf{z}) = \mbb{P} \mbf{u} +  \mbb{P} \mbf{e} \triangleq \tilde{\mbf{u}} + \tilde{\mbf{e}} \,.
\end{equation}
Ultimately, we evaluate the quality of the denoised estimates using the reconstruction error, which is defined as
\begin{equation} \label{eq:modelingError}
\tilde{E}(f(\mbf{z}), \mbf{u}) \triangleq  \left \| f(\mbf{z}) - \mbf{u} \right \|_2^2 = \left \| (\tilde{\mbf{u}} - \mbf{u}) + \tilde{\mbf{e}} \right \|_2^2 \, .
\end{equation}
Since the initial noise $\mbf{e}$ is uncorrelated with the image data $\mbf{u}$,
and $\tilde{\mbf{u}} - \mbf{u} = (\mbf{I} - \mbb{P}) \mbf{u}$ is orthogonal to $\tilde{\mbf{e}} = \mbb{P} \mbf{e}$, 
the reconstruction error is equivalent to
\begin{align} \label{eq:decouple}
\nonumber \tilde{E}(f(\mbf{z}), \mbf{u}) & = \left \| \tilde{\mbf{u}} - \mbf{u} \right \|_2^2 + \left \| \tilde{\mbf{e}} \right \|_2^2 \\
& \triangleq \tilde{E}_m + \tilde{E}_n \;,
\end{align}
where $\tilde{E}_m = \left \| \tilde{\mbf{u}} - \mbf{u} \right \|_2^2$ denotes the data \textit{modeling error}, and $\tilde{E}_n = \left \| \tilde{\mbf{e}} \right \|_2^2$ denotes the \textit{survived noise energy}.

To reduce the reconstruction error, one needs to $(i)$ preserve $\mbf{u}$ with small $\tilde{E}_m$, and $(ii)$ remove as much noise as possible to minimize $\tilde{E}_n$. 
Overall, the goal of denoising algorithms is to maximize the signal-to-noise ratio (SNR) of the denoised estimate $ \left \| \mbf{u} \right \|_2^2 / (\tilde{E}_m + \tilde{E}_n)$. 
To simplify the analysis, we investigate the denoising of image patches from an image corpus, which are denoted as $\begin{Bmatrix} \mbf{u}_i \end{Bmatrix}_{i \in \Omega}$.
We evaluate the \textit{normalized modeling error} $\alpha$ and \textit{survived noise energy ratio} $\beta$, which are defined as
\begin{equation} \label{eq:empirical}
\alpha \triangleq \frac{\sum_{i \in \Omega} \left \| \tilde{\mbf{u}}_i - \mbf{u}_i \right \|_2^2}{\sum_{i \in \Omega} \left \| \mbf{u}_i \right \|_2^2}, \;\;\; 
\beta \triangleq \frac{\sum_{i \in \Omega} \left \| \tilde{\mbf{e}}_i \right \|_2^2}{\sum_{i \in \Omega} \left \| \mbf{e}_i \right \|_2^2} \, .
\end{equation}
Eventually, we evaluate the effectiveness of certain image model using the empirical SNR of the denoised output
\begin{equation} \label{eq:SNR}
\text{SNR}_{\text{out}} \triangleq  \frac{\sum_{i \in \Omega} \left \| \mbf{u}_i \right \|_2^2}{\sum_{i \in \Omega} \left \| \tilde{\mbf{u}}_i + \tilde{\mbf{e}}_i - \mbf{u}_i \right \|_2^2} = \frac{1}{\alpha + \beta \, / \, \text{SNR}_{\text{in}}} \, .
\end{equation}
Here $\text{SNR}_{\text{in}}$ denotes the input SNR.
The proposed metric implies that, minimizing $\beta$ becomes more important comparing to $\alpha$ as $\text{SNR}_{\text{in}}$ becomes smaller.

\section{Experiments} \label{sec5}
We conduct various experiments to study the effectiveness of image modeling.
We first denoise image patches, which are generated from an image corpus, by projecting them onto the solution set of a single image model, and combining multiple models using methods described in Section \ref{sec42} and Section \ref{sec43}, respectively.
Furthermore, based on the scheme described in Section \ref{sec44}, we evaluate the denoised results by combining several popular image denoising algorithms, based on different image models.
Last but not least, we study what properties the deep denoising neural network exploits inexplicitly, by applying the proposed scheme in Section \ref{sec44} as a computational tool for evaluation.
We show improved denoising results over those from the state-of-the-art denoising network, by combining deep leraning with algorithms which exploit image properties that the learned neural networks fail to capture.

\subsection{Image Patch Denoising} \label{sec51}

Individual images can have very distinct structures and properties, thus evaluation of the denoised patch results from a single image may favor algorithms based on specific image models.
Here, we work with the \textit{Kodak} image dataset \cite{kodak} as our image corpus, which contains $24$ lossless images with diverse features.
The true color images are first converted to gray-scale.
From each image, we randomly select $1000$ patches of size $8 \times 8$ (thus, there are $N = 24000$ selected patches in total).              
We set up a $50 \times 50$ search window centered at each selected patch, and find its $M$ nearest neighbors within the windows \cite{wen2017joint,wen2017sparsity}, where $M = 64$. 
The $M$ patches in the $i$-th group are vectorized and form the columns of $\Xb_i \in \Rmb^{n \times M}$, and thus $\begin{Bmatrix} \Xb_i \end{Bmatrix}_{i=1}^{N}$ are the ground-true image data for our denoising experiment. 
We simulate i.i.d. Gaussian noise matrices $\begin{Bmatrix} \mbf{\Phi}_i \end{Bmatrix}_{i=1}^{N}$ of standard normal distribution (i.e., zero mean and unit standard deviation), which have same size as $\Xb_i$'s.
Thus the noisy data $\begin{Bmatrix} \Zb_i \end{Bmatrix}_{i=1}^{N}$ with noise standard deviation $\sigma$ are generated as $\Zb_i = \Xb_i + \sigma \mbf{\Phi}_i$ $\forall i$.

\subsubsection{Single Model}

We denoise $\begin{Bmatrix} \Zb_i \end{Bmatrix}_{i=1}^{N}$ by projecting them onto the set of SP, LR, JS, and GS models, respectively.
We evaluate the denoised estimates by plotting the $\alpha$ and $\beta$ against $K$, i.e., the sparsity level or rank number in the models, with fixed $\sigma = 20$.
Furthermore, we show the plot of $\text{SNR}_{\text{out}}$ against the $\text{SNR}_{\text{in}}$ with fixed $K = 10$, to illustrate the quality improvement of the denoised estimates, at different noise levels of the input data.

We first conduct \textit{oracle test}, meaning that the models are trained, and the projection operator is determined using the clean data $\begin{Bmatrix} \Xb_i \end{Bmatrix}_{i=1}^{N}$, which excludes the noise overfitting.
Fig. \ref{fig:oraclePlot}(a) plots the normalized modeling error $\alpha$ against the $K$ value in the oracle test.
The most flexible $GS$ model, and the most restrictive $JS$ model lead to the smallest, and the highest modeling error $\alpha$, respectively. 
Such empirical results are in accord with the theoretical analysis in Section \ref{sec42}.
Since the learned dictionaries or subspaces are unitary, and are uncorrelated with the noise, the noise is distributed uniformly in all bases. 
The $\beta$ plots are all linear against $K$, and identical for all models.
Thus the improvement of $\text{SNR}_{\text{out}}$ over $\text{SNR}_{\text{in}}$ only depends on $\alpha$.
Fig. \ref{fig:oraclePlot}(b) plots the output SNR against the input SNR, in which the $GS$ model based denoising provides largest quality improvement.

\begin{figure}[!t]
\begin{center}
\begin{tabular}{cc}
\includegraphics[height=1.1in]{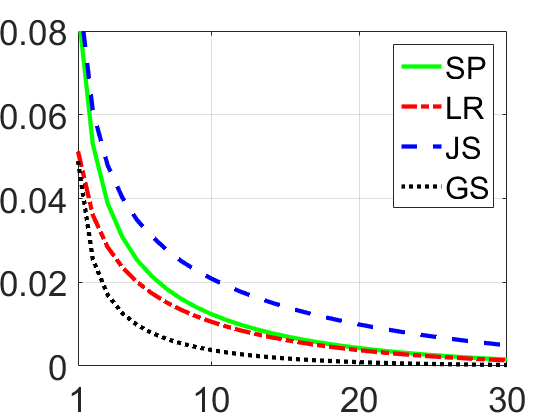} &
\includegraphics[height=1.1in]{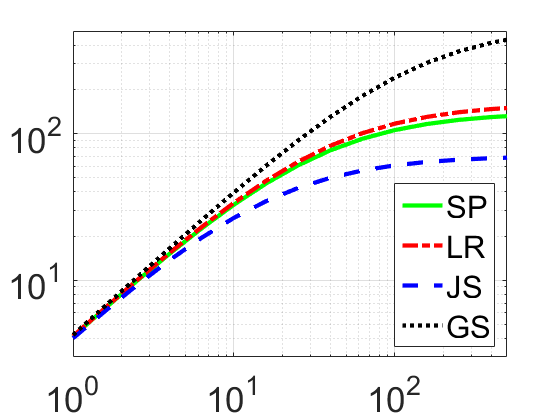}\\
{\small (a) Plot of $\alpha$ v.s $K$. } &  {\small (b) Plot of $\text{SNR}_{\text{out}}$ v.s $\text{SNR}_{\text{in}}$.}
\end{tabular}
\caption{Oracle denoising tests using single image model: (a) Plot of the $\alpha$ value against the $K$ when fixing $\sigma = 20$, and (b) Plot of the output SNR against the input SNR when fixing $K = 10$.}
\label{fig:oraclePlot}
\end{center}
\vspace{-0.12in}
\end{figure}

In practice, most of the popular model-based image restoration algorithms proposed to learn the sparse or low-rank models using the corrupted measurements directly.
Thus, the trained dictionary, or the low-dimensional subspace may be overfitting to the noise, depending on noisy variance, as well as the model complexity, for a fixed set of training samples.
To study the influence of the noise, we train the models using the noisy measurements $\begin{Bmatrix} \Zb_i \end{Bmatrix}_{i=1}^{N}$ with noise $\sigma = 20$, and conduct the denoising test.
Fig. \ref{fig:noisyPlot}(a) plots the normalized modeling error $\alpha$ against the $K$ value.
Comparing to Fig. \ref{fig:oraclePlot}(a), the $\alpha$ value of the denoised estimate using the GS model becomes smaller, relative to the results using other models, especially when $K$ is small.
Whereas the restrictive $JS$ models now provides high $\alpha$.
Fig. \ref{fig:noisyPlot}(b) plots the remaining noise ratio $\beta$, which is no longer identical for all models.
The $GS$ model leads to much higher $\beta$ due to noise overfitting in training, while the $JS$ model is relatively more robust to noise.
Fig. \ref{fig:noisyPlot}(c) plots the $\text{SNR}_{\text{out}}$ against the $\text{SNR}_{\text{in}}$. 
Different from Fig. \ref{fig:oraclePlot}(b) in which the $GS$ model always provides best denoised result, the restrictive JS model provides the best denoised estimates with the highest $\text{SNR}_{\text{out}}$ among all models, when $\text{SNR}_{\text{in}}$ is low (i.e. image is noisy).
As $\text{SNR}_{\text{in}}$ keeps increasing, the $\text{SNR}$ plot in Fig. \ref{fig:noisyPlot}(c) converges to the oracle results in Fig. \ref{fig:oraclePlot}(b).
Furthermore, instead of using fixed $K$ values, we search for the optimal $K$ which generate the highest $\text{SNR}_{\text{out}}$ for each $\text{SNR}_{\text{in}}$ using each model \footnote{The popular denoising algorithms usually have specific approaches of selecting $K$~\cite{Mairal2009,gu2017weighted,octobos}}.
Fig.~\ref{fig:bestKplot}(a) plots the $\text{SNR}_{\text{out}}$ against the $\text{SNR}_{\text{in}}$ using the optimal $K$'s, which demonstrate similar behavior comparing Fig. \ref{fig:noisyPlot}(c).

\begin{figure}[!t]
\begin{center}
\begin{tabular}{ccc}
\includegraphics[height=1.05in]{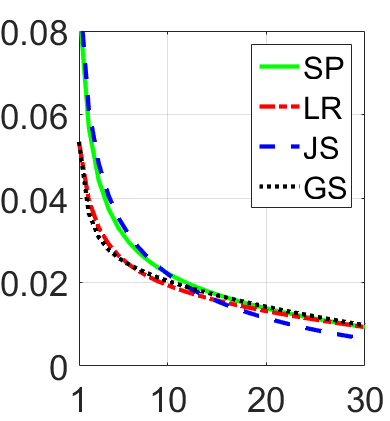} &
\includegraphics[height=1.05in]{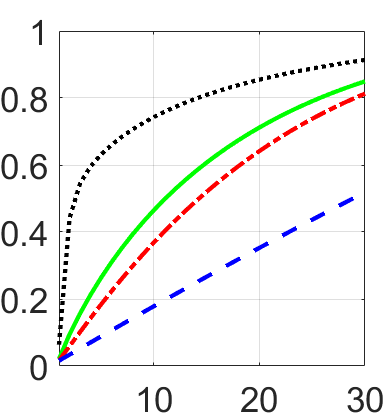} &
\includegraphics[height=1.05in]{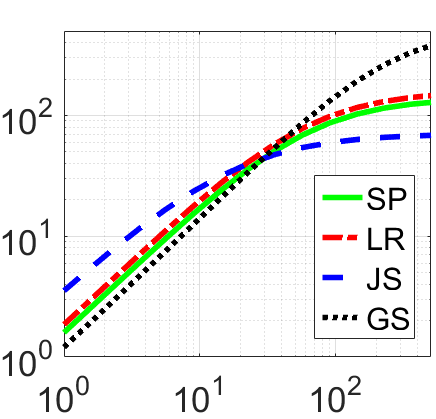}\\
{\small (a) $\alpha$ v.s $K$} &  {\small (b) $\beta$ v.s $K$} & {\small (c) $\text{SNR}_{\text{out}}$ v.s $\text{SNR}_{\text{in}}$}
\end{tabular}
\vspace{-0.07in}
\caption{Denoising test using single image model: (a) Plot of the $\alpha$ value against the $K$, (b) Plot of the $\beta$ value against the $K$ when fixing $\sigma = 20$, and (c) Plot of the output SNR against the input SNR when fixing $K = 10$.}
\label{fig:noisyPlot}
\end{center}
\vspace{-0.2in}
\end{figure}

\begin{figure}
\begin{center}
\begin{tabular}{cc}
\includegraphics[height=1.1in]{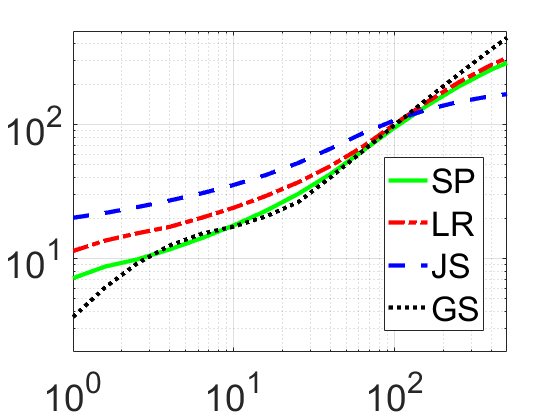} &
\includegraphics[height=1.1in]{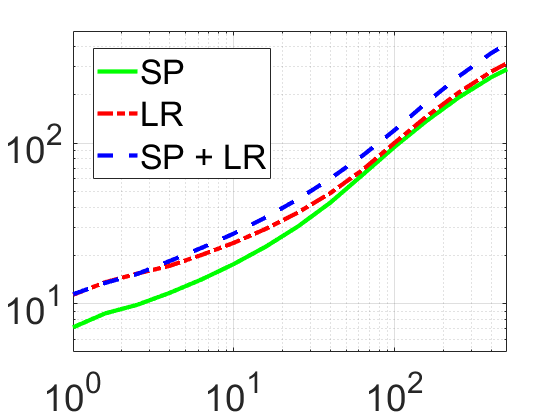}\\
{\small (a) $\text{SNR}$s by single models. } &  {\small (b) Combining SP and LR.}
\end{tabular}
\caption{The plot of output $\text{SNR}$s against input $\text{SNR}$s with the optimal $K$'s: Denoising results by (a) using single image models, and (b) combing SP and LR models.}
\label{fig:bestKplot}
\end{center}
\vspace{-0.12in}
\end{figure}

\subsubsection{Multiple Models}
When applying multiple image models, Section \ref{sec43} provides some intuitions why alternating projection algorithms may fail to generate accurate estimates.
Besides, we conduct dual-model denoising test, and compare the empirical results by applying each single model, and two image models jointly, using both the alternating projection and the convex combination algorithms described in Section \ref{sec43}.
For the convex combination algorithm, I set the weighting factor $\mu = 0.5$ \footnote{There exists optimal $\mu$ which may further improve the denoised estimate. Here we naively set $\mu = 0.5$ which has already showed improvement over the competing approaches.}.

We show the denoising results combining the SP and LR models as an example.
Fig. \ref{fig:alternating}(a) plots the value of $\alpha$ and $\beta$ against $K$, using the SP and LR models independently, as well as using both via alternating projection, and convex combination methods.
It is clear that the denoised estimate using alternating projection algorithm yields larger modeling error, comparing to those obtained by simple projection onto the solution set of a single model.
Whereas the results by applying convex combination has smaller modeling error, which improves the denoising quality.
Furthermore, Fig. \ref{fig:alternating}(b) shows that the alternating projection algorithm generates much larger remaining noise, while convex combination approach can effectively suppress the $\beta$ value.
We also observe similar behavior when denoising image data by combining other image models.
Therefore, the convex combination turns out to be an effective method for jointly imposing multiple image models in denoising.

As the convex combination of results using multiple image models improves the denoising performance, we study what the best combination is to enhance the quality of the denosied estimates.
We provide empirical results which support our Conjecture \ref{conj:complement} in patch denoising.
Fig. \ref{fig:combinationImprove} (a) plots the $\text{SNR}_{\text{out}}$ of the denoised estimates against the $\text{SNR}_{\text{in}}$ of the noisy images, using algorithms based on the LR and GS models.
The convex combination of the results using these two single models provides marginal performance improvement.
It is in accord with the relationship of the solution sets shown in Fig. \ref{fig:Relationship}. Since the LR set is contained in the GS set, their intersection is relatively large.
Fig. \ref{fig:combinationImprove} (b) plots the $\text{SNR}_{\text{out}}$ of the denoised estimates using algorithms based on the LR and SP models.
Different from the LR set and the GS set, Fig. \ref{fig:Relationship} indicates that the SP set has small intersection with LR set.
As a result, larger improvement is observed when combining the results based on the SP model, and the LR model.
Instead of applying fixed $K$ values, we select the optimal $K$ for the denoising test using LR and SP models.
Fig.~\ref{fig:bestKplot} (b) plots the $\text{SNR}_{\text{out}}$ of the denoised estimates against $\text{SNR}_{\text{in}}$. 
The convex combination of results using SP and LR models clearly outperform the algorithm using single model.

\begin{figure}[!t]
\begin{center}
\begin{tabular}{cc}
\includegraphics[height=1.1in]{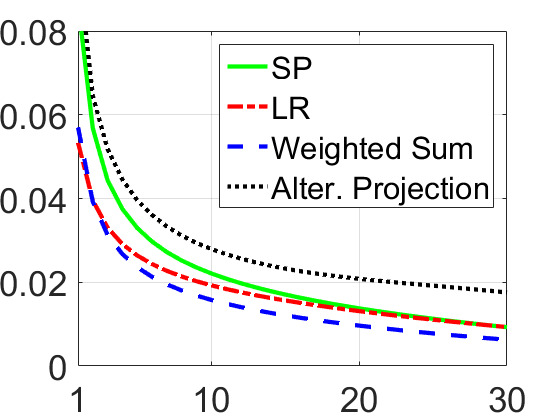} &
\includegraphics[height=1.1in]{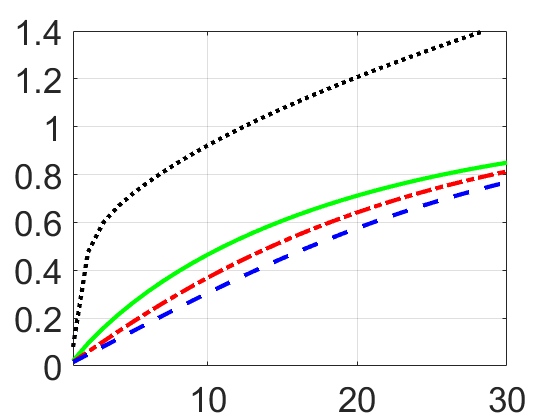}\\
{\small (a) Plot of $\alpha$ v.s $K$. } &  {\small (b) Plot of  $\beta$ v.s $K$.}
\end{tabular}
\caption{Comparison of the denoising results using single model, and multiple models by alternating projection and convex combination: (a) Plot of the $\alpha$ value against the $K$, and (b) Plot of the  $\beta$ value against the $K$.}
\label{fig:alternating}
\end{center}
\vspace{-0.12in}
\end{figure}

\begin{figure}[!t]
\begin{center}
\begin{tabular}{cc}
\includegraphics[height=1.1in]{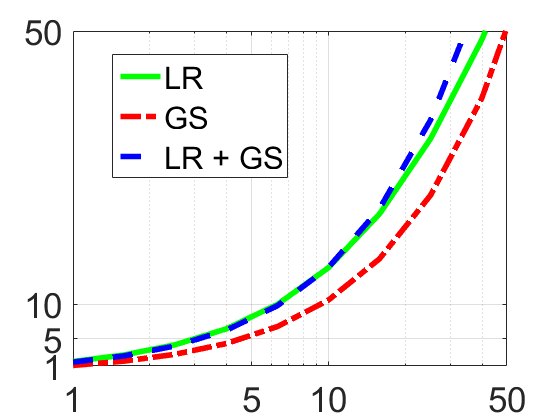} &
\includegraphics[height=1.1in]{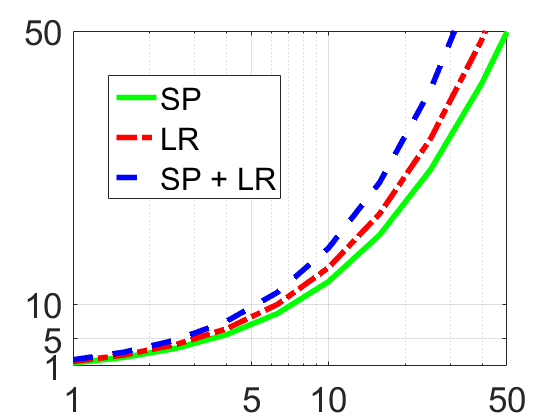}\\
{\small (a) Combining LR and GS. } &  {\small (b) Combining LR and SP.}
\end{tabular}
\caption{The plot $\text{SNR}_{\text{out}}$ using single models, and their convex combination: (a) Plot of the $\alpha$ value against the $K$, and (b) Plot of the  $\beta$ value against the $K$.}
\label{fig:combinationImprove}
\end{center}
\vspace{-0.12in}
\end{figure}

\subsection{Model-based Image Denoising} \label{sec52}

Various model-based image restoration methods have been proposed recently, which achieved promising performance in image denoising.
Comparing to the simple denoising methods by projection which we introduced in Section \ref{sec42}, the popular image denoising algorithms are usually more complicated which involves additional steps, including patch aggregation, block matching, applying special shrinkage function, etc.
However, the core of these algorithms are still based on solution set projection.
Here, we provide numerical results to show that the Conjecture \ref{conj:complement} also holds for the convex combination of results using multiple image denoising algorithms, which is described in Section \ref{sec44}.

We select popular image denoising algorithms based on image models that we analyzed, including
\begin{itemize}
	\item Sparsity (SP) model: KSVD \cite{elad}, and OCTOBOS \cite{octobos},
	\item Low-Rank (LR) model: SAIST \cite{Dong2013}, and WNNM \cite{gu2017weighted},
	\item Group-wise Sparsity (GS) model: SSC-GSM \cite{dong2015image}.
\end{itemize}
The publicly available codes from their authors' websites are used for implementation of the image denoising tests.
We use the $24$ lossless images (converted to gray-scale) from the \textit{Kodak} image dataset \cite{kodak} as the testing images, and simulate i.i.d. Gaussian noise at 4 different noise levels ($\sigma = 5, 10, 15$ and $25$) to generate the noisy images.
The images are denoised using the selected popular restoration algorithms.
We apply the Peak Signal-to-Noise Ratio (PSNR) in decible (dB) as the objective metric to evaluate the quality of the denoised images.
The denoised results using the selected pairs of algorithms are combined.
We set $\lambda_f = 1e-2$, and conduct a line search to use the best weight $\mu$ between $0$ and $1$, which provides the highest PSNR of the combined result.
The reported PSNR value is averaged over the $24$ testing images, for each noise $\sigma$ and method.

\begin{figure*}
\begin{center}
\begin{tabular}{ccccc}
\includegraphics[height=1.24in]{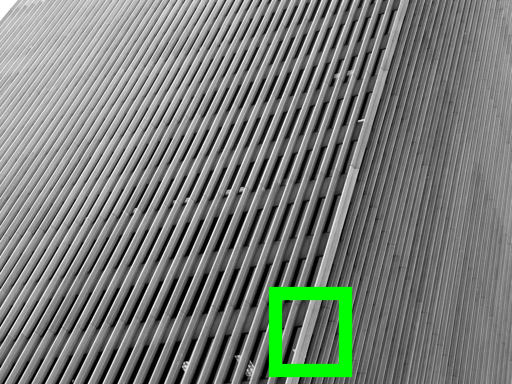} &
\includegraphics[height=1.24in]{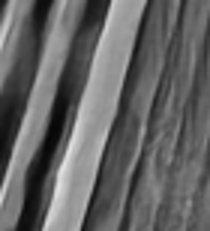} &
\includegraphics[height=1.24in]{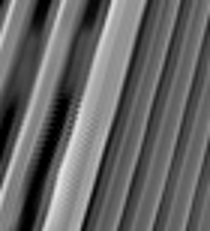} & 
\includegraphics[height=1.24in]{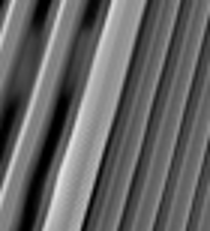} &
\includegraphics[height=1.24in]{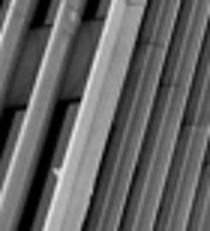}
\\
\includegraphics[height=1.15in]{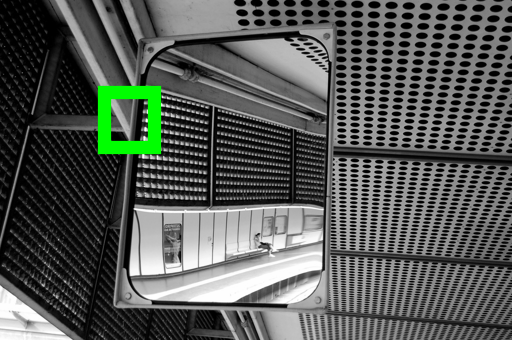} &
\includegraphics[height=1.24in]{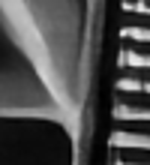} &
\includegraphics[height=1.24in]{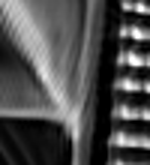} & 
\includegraphics[height=1.24in]{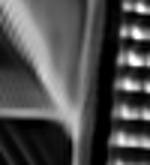} &
\includegraphics[height=1.24in]{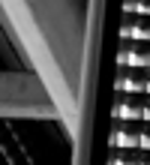}
\\
\includegraphics[height=1.15in]{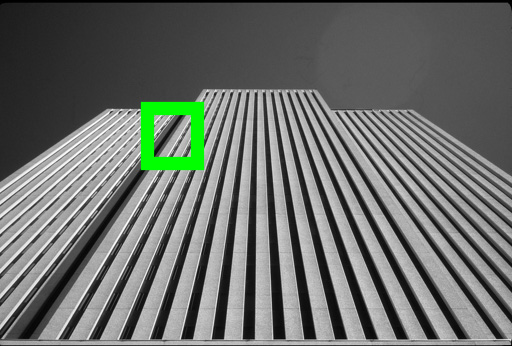} &
\includegraphics[height=1.24in]{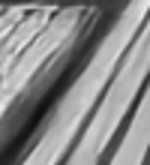} &
\includegraphics[height=1.24in]{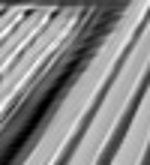} & 
\includegraphics[height=1.24in]{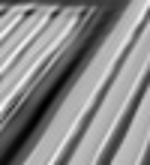} &
\includegraphics[height=1.24in]{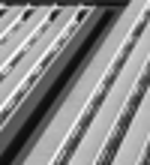}
\\
{\footnotesize (a) Ground Truth with Zoom-in} & {\footnotesize (b) DnCNN} & {\footnotesize (c) WNNM} & {\footnotesize (d) Proposed}  & {\footnotesize (e) Ground Truth}
\label{fig:denoisingVisual}
\end{tabular}
\end{center}
\vspace{-0.1in}
\caption{Denoising results of (a) the example images from the Urban100 Dataset \cite{huang2015single} at $\sigma = 70$, with the green rectangles highlighting the zoom-in regions of (b) the denoised images by DnCNN (PSNR = $22.26$ / $21.08$ / $27.13$ dB), and (c) the denoised images by WNNM (PSNR = $26.15$ / $24.10$ / $29.28$ dB), and (d) the denoised images by the proposed convex combination (PSNR = $\mathbf{27.06}$ / $\mathbf{24.77}$ / $\mathbf{30.21}$ dB).}
\vspace{-0.1in}
\end{figure*}

Table \ref{table:im_denoise_1} lists the average PSNRs of the denoised images using KSVD, OCTOBOS, and their convex combination.
Since both KSVD and OCTOBOS are image denoising algorithms based on the SP model, the convex combination (with the optimal factor $\mu$) of their estimates only provide marginal improvement.

\begin{table}[h!]
\centering
\fontsize{9}{14pt}\selectfont
\begin{tabular}{||r | c c c c||} 
 \hline
 & $\sigma$ =5  & $\sigma$ =15  & $\sigma$ =25 & Average\\
 \hline\hline
KSVD  & 37.60  & 31.59  & 29.12 & 32.77 \\
OCTOBOS  & 38.27   & 32.16  & 29.60 & 33.34 \\
KSVD + OCTOBOS  & 38.28   & 32.20  & 29.64 & 33.37\\
\hline
$\Delta$ PSNR mean& 0.01   & 0.04  & 0.04  & 0.03\\
$\Delta$ PSNR std & 0.01   & 0.02  & 0.02 & 0.02\\
 \hline
\end{tabular}
\caption{Average PSNR of the denoised images by KSVD, OCTOBOS, their convex combination, and the improved PSNR with its standard deviation (std) at different noise levels}
\label{table:im_denoise_1}
\end{table}

\begin{table}[h!]
\centering
\fontsize{9}{14pt}\selectfont
\begin{tabular}{||r | c c c c||} 
 \hline
 & $\sigma$ =5  & $\sigma$ =15  & $\sigma$ =25 & Average\\
 \hline\hline
SSC-GSM  & 38.37  & 32.37  & 29.76 & 33.50 \\
SAIST  & 38.39   & 32.39  & 29.93 & 33.57 \\
SSC-GSM + SAIST  & 38.46  & 32.49  & 29.99 & 33.64\\
\hline
$\Delta$ PSNR & 0.06  & 0.07  & 0.05 & 0.07\\
$\Delta$ PSNR std & 0.02  & 0.03  & 0.04 & 0.03\\
 \hline
\end{tabular}
\caption{Average PSNR of the denoised images by SSC-GSM, SAIST, their combination, and the improved PSNR with its std under different noise levels}
\label{table:im_denoise_2}
\end{table}

\begin{table}[h!]
\centering
\fontsize{9}{14pt}\selectfont
\begin{tabular}{|| r | c c c c||} 
 \hline
 & $\sigma$ =5  & $\sigma$ =15 & $\sigma$ =25 & Average \\
 \hline\hline
OCTOBOS  & 38.27   & 32.16  & 29.60  & 33.34 \\
SAIST  & 38.39    & 32.39  & 29.93  & 33.57 \\
OCTOBOS + SAIST  & 38.50   & 32.50  & 30.02 & 33.67\\
\hline
$\Delta$ PSNR mean& 0.10  & 0.11  & 0.09 & 0.10\\
$\Delta$ PSNR std & 0.03  & 0.04  & 0.04 & 0.04\\
 \hline
\end{tabular}
\caption{Average PSNR of the denoised images by octobos, saist, their combination and improved PSNR (mean and std) under different noise levels}
\label{table:im_denoise_3}
\end{table}

Table \ref{table:im_denoise_2} lists the average PSNRs of the denoised images using SSC-GSM, SAIST, and their convex combination.
The SSC-GSM algorithm is based on the GS model, while SAIST is based on the LR model.
They are different image models based on our analysis, thus the convex combination of their estimates provides relatively larger improvement.

Table \ref{table:im_denoise_3} lists the average PSNRs of the denoised images using OCTOBOS, SAIST, and their convex combination.
The image patch denoising results in Section \ref{sec52} demonstrates that jointly imposing the SP and LR models can effectively improve the denoised estimates.
Here, we observe the similar results: the convex combination of the image denoising algorithms based on the LR and SP models provides more PSNR improvement over other combinations.
It is in accord with Conjecture \ref{conj:complement} in the image denoising experiments.

\subsection{Understand and Enhance Deep Neural Networks for Image Denoising} \label{sec53}

Besides the model-based image restoration algorithms, recent works applied the popular deep learning technique in various inverse problems which showed promising performance.
The recently proposed DnCNN \cite{zhang2017beyond} demonstrated superior image denoising results comparing to the model-based methods.
Different from conventional approaches solving inverse problems, the deep learning approach requires a large training corpus, and has little assumption on the image priors.
However, it is unclear what image properties and models the learned neural network exploits.

We apply the convex combination approach for image denoising, as a computational tool to study the relationship between the learned neural networks and the well-defined image models.
With the same image denoising setup in Section \ref{sec52}, the $24$ images with simulated i.i.d. Gaussian noise at $\sigma = 30, 50$ and $70$ (for which the released DnCNN models have corresponding $\sigma$ levels), from Kodak set are denoised using the trained DnCNN networks, which are available from the authors' GitHub repository \cite{dncnn}.
We combine the denoised estimates using DnCNN, and other image denoising algorithms using different image models. 
We set $\lambda_f = 10^{-2}$, and use the best weight $\mu$ between $0$ and $1$ to achieve the highest PSNR of the combined results.
Based on Conjecture \ref{conj:complement}, if the combined estimate fails provide PSNR improvement over the results using DnCNN or model-based algorithm alone, such image model has been exploited inexplicitly by the learned neural networks.
On the other hand, if such combination can further improve the denoising performance using single method, the corresponding image model has not been fully exploited by deep learning.

\begin{table}[h!]
\centering
\fontsize{9}{14pt}\selectfont
\begin{tabular}{|| r | c c c | c||} 
 \hline
 & $\sigma$ =15  & $\sigma$ =25  & $\sigma$ =50  & $\Delta $ PSNR\\
 \hline\hline
DnCNN  & 32.89  & 30.47  & 27.49 &  0.00 \\
DnCNN + KSVD  & 32.89  & 30.47  & 27.49 & 0.00\\
DnCNN + SSC-GSM  & 32.94  & 30.52  & 27.58 & 0.06\\
DnCNN + WNNM  & 32.95  & 30.55  & 27.60 & 0.08\\ \hline
\end{tabular}
\caption{Average PSNR of the denoised images from Kodak Set by DnCNN, and its combination with other model-based image denoising methods, with their corresponding improved PSNRs.}
\label{table:im_denoise_4}
\end{table}

\begin{table}[h!]
\centering
\fontsize{9}{14pt}\selectfont
\begin{tabular}{|| r | c c c | c||} 
 \hline
 & $\sigma$ =30  & $\sigma$ =50  & $\sigma$ =70  & $\Delta $ PSNR\\
 \hline\hline
DnCNN  & 28.16  & 25.48  & 23.61 & 0.00\\
DnCNN + KSVD  & 27.94  & 25.20  & 23.26 & -0.29\\
DnCNN + SSC-GSM  & 28.64  & 25.97  & 24.17 & 0.51\\
DnCNN + WNNM  & \textbf{28.82}  & \textbf{26.09}  & \textbf{24.26} & \textbf{0.64}\\
\hline
\end{tabular}
\caption{Average PSNR of the denoised images from the Urban100 Dataset \cite{huang2015single} by DnCNN, and its combination with KSVD, SSC-GSM, and WNNM, with their corresponding improved PSNRs. The highest PSNRs for each $sigma$ and the highest $\Delta$ PSNR are highlighted in bold.}
\label{table:im_denoise_5}
\end{table}

Fig. \ref{table:im_denoise_4} lists the average PSNRs of the denoised images using DnCNN, as well as those are convex combination with model-based algorithms, including KSVD, SSC-GSM, and WNNM, at each testing noise level.
For the results that combines DnCNN and other methods, the PSNR improvement comparing to those using DnCNN alone is listed as $\Delta$ PSNR.
The results using DnCNN$+$KSVD do not provide any PSNR improvement, even with the best possible weight $\mu$.
Such empirical results show that the learned DnCNN network has exploited image local sparsity inexplicitly.
On the contrary, the results using DnCNN$+$SSC-GSM and DnCNN$+$WNNM provide even higher PSNR than the result using only the state-of-the-art DnCNN network.
It demonstrated that the learned DnCNN network has not fully captured non-local image properties, such as group-wise sparsity and low-rankness, which are closely related to the well-known image self-similarity.

The aforementioned experiments assume that the optimal $\mu$ is known, for which the value search requires the oracle to be available.
However, in practice, the weight $\mu$ needs to be either adaptively learned, or fixed.
Now we demonstrate how much the simple convex combination approach can enhance the image denoising results over state-of-the-art deep learning method, by naively fixing $\mu = 0.5$.
To verify that the conjecture that the image self-similarity is not fully exploited by the learned DnCNN network, we work with the $100$ images (converted to gray-scale) from the Urban100 dataset \cite{huang2015single} which contains image having repeating structures, and simulate i.i.d. Gaussian noise at $\sigma = 30, 50$ and $70$ to generate the noisy testing images.

Fig. \ref{table:im_denoise_5} lists the average PSNRs of the denoised images using the learned DnCNN network, as well as its combination with KSVD, SSC-GSM, and WNNM, at each testing noise level.
The PSNR improvement by each approach comparing to those using DnCNN alone is listed as $\Delta$ PSNR.
With the fixed weight $\mu = 0.5$, the method combining DnCNN with KSVD provides lower denoised PSNR, which is accord with the previous analysis.
Whereas the DnCNN results combined with SSC-GSM and WNNM achieve noticeable quality improvement, with $\Delta$ PSNR $= 0.51$ and $0.64$, respectively.
It provides confidence and promise on further improvement over the current state-of-the-art image restoration neural networks.
Fig. \ref{fig:denoisingVisual} compares the denoised images using DnCNN, and WNNM alone, to the denoised results by combination of the two.
The results by DnCNN usually recover image details, while introducing spatial distortion.
On the contrary, the results using WNNM have smooth spatial structures, but contain undesired artifacts.
The results by convex combination of the two achieve quality improvement in terms of both denoised PSNR, and the visual quality.

\section{Conclusion} \label{sec6}

We provide theoretical analysis on image models that are used in popular image restoration algorithms.
The relationship among the solution sets of sparsity, group-wise sparsity, joint sparsity, and low-rankness models are presented and proved, with mild assumptions.
We propose objective metrics to evaluate how effective each of these image models are applied in image denoising algorithm.
When images are denoised via weighted combination of results by projection onto the solution sets of single models, we provide a condition which guarantee the image quality improvement in terms of SNR.
It turns out that the combination of complementary image models provides larger denoising performance improvement, and we supply empirical evidence which supports our conjecture.
Furthermore, we apply the proposed denoising framework by weighted combination to study the image properties that are exploited by deep learnng.
We show that the denoised results using the state-of-the-art deep learning methods can be further improved by the proposed framework.
With the knowledge and understanding of the relationship of image models, we plan to develop more advanced image restoration scheme by applying multiple effective regularizers.

\ifCLASSOPTIONcaptionsoff
  \newpage
\fi

\bibliographystyle{./IEEEtran}
\bibliography{./refs}

\end{document}